\documentclass[11pt]{amsart}

\usepackage{etex}
\reserveinserts{28}

\usepackage[margin=1in]{geometry}

\usepackage{graphicx,latexsym,amsmath,amsthm,amssymb,color,verbatim,mathrsfs,bbm,amscd}
\usepackage[latin1] {inputenc}
\usepackage[english]{babel}
\usepackage{setspace}
\usepackage[vcentermath]{youngtab}
\usepackage{ytableau}
\usepackage{leftidx}
\usepackage[all,cmtip]{xy}
\usepackage{tikz}
\usepackage{url}
\usepackage{paralist}

\usepackage{cite}

\usetikzlibrary{matrix,decorations.pathreplacing}
\usepackage{floatrow}
\newcommand{\Sym}{\mathsf{Sym}}

\newcommand{\HWV}{\mathsf{HWV}}
\newcommand{\dotitem}{\item[$\cdot$]}
\DeclareMathOperator{\Hom}{Hom}
\newcommand{\Pow}{\mathrm{Pow}}
\DeclareMathOperator{\Aut}{Aut}
\DeclareMathOperator{\Inn}{Inn}
\DeclareMathOperator{\Out}{Out}
\newcommand{\bbC}{\mathbb{C}}
\newcommand{\fraksl}{\mathfrak{sl}}

\newcommand{\vvirg}{ , \dots , }

\renewcommand{\Im}{\mathrm{Im}}

\newcommand{\la}{\lambda}
\newcommand{\IC}{\ensuremath{\mathbb{C}}}
\newcommand{\IA}{\ensuremath{\mathbb{A}}}
\newcommand{\IN}{\ensuremath{\mathbb{N}}}
\newcommand{\IZ}{\ensuremath{\mathbb{Z}}}

\newcommand{\aS}{\ensuremath{\mathfrak{S}}}

\newcommand{\bbZ}{\mathbb{Z}}
\newcommand{\bbS}{\mathbb{S}}
\newcommand{\frakS}{\mathfrak{S}}

\newcommand{\pc}{\textup{\textsf{pc}}}

\newcommand{\dc}{\textup{\textsf{dc}}}

\newcommand{\per}{\mathrm{per}}
\renewcommand{\det}{\mathrm{det}}
\newcommand{\sgn}{\mathrm{sgn}}
\newcommand{\VNP}{\textbf{VNP}}
\newcommand{\VPs}{\textbf{VP$_{\text{s}}$}}
\newcommand{\VPe}{\textbf{VP$_{\text{e}}$}}
\newcommand{\tr}{\mathrm{tr}}

\newcommand{\calS}{\mathcal{S}}
\newcommand{\frakg}{\mathfrak{g}}

\newcommand{\textsum}{{\textstyle \sum}}

\DeclareMathOperator{\End}{End}

\newcommand{\xto}{\xrightarrow}

\newcommand{\tensor}{\smash{\textstyle\bigotimes}}
\newcommand{\GL}{\mathsf{GL}}
\newcommand{\PGL}{\mathsf{PGL}}
\newcommand{\SL}{\mathsf{SL}}
\newcommand{\Id}{\mathrm{Id}}

\swapnumbers
\numberwithin{equation}{section}
\newtheorem{theorem}[equation]{Theorem}
\newtheorem{corollary}[equation]{Corollary}
\newtheorem{lemma}[equation]{Lemma}

\newtheorem{proposition}[equation]{Proposition}
\newtheorem{conjecture}[equation]{Conjecture}

\theoremstyle{definition}
\newtheorem{observ}[equation]{Observation}
\newtheorem{remark}[equation]{Remark}

\newcommand{\partinto}[1][]{\smash{\mathord{\mathchoice{%
  \xymatrix@=0.4em@1{%
  \ar@{|-}[rr]_-*--{\scriptstyle #1}
  &*{\phantom{\scriptstyle{#1}}}&}
}{
  \xymatrix@=0.25em@1{%
  \ar@{|-}[rr]_-*--{\scriptstyle #1}
  &*{\phantom{\scriptstyle{#1}}}&}
}{
  \xymatrix@=0.2em@1{%
  \ar@{|-}[rr]_-*--{\scriptscriptstyle #1}
  &*{\phantom{\scriptscriptstyle{#1}}}&}
}{}}}}
\newcommand{\partintonosmash}[1][]{\mathord{\mathchoice{%
  \xymatrix@=0.4em@1{%
  \ar@{|-}[rr]_-*--{\scriptstyle #1}
  &*{\phantom{\scriptstyle{#1}}}&}
}{
  \xymatrix@=0.25em@1{%
  \ar@{|-}[rr]_-*--{\scriptstyle #1}
  &*{\phantom{\scriptstyle{#1}}}&}
}{
  \xymatrix@=0.2em@1{%
  \ar@{|-}[rr]_-*--{\scriptscriptstyle #1}
  &*{\phantom{\scriptscriptstyle{#1}}}&}
}{}}}
\newcommand{\partintostar}[1][]{\smash{\mathord{\mathchoice{%
  \xymatrix@=0.4em@1{%
  \ar@{|-}[rr]_-*--{\scriptstyle #1}^-*--{\scriptstyle \ast}
  &*{\phantom{\scriptstyle{#1}}}&}
}{
  \xymatrix@=0.25em@1{%
  \ar@{|-}[rr]_-*--{\scriptstyle #1}^-*--{\scriptstyle \ast}
  &*{\phantom{\scriptstyle{#1}}}&}
}{
  \xymatrix@=0.2em@1{%
  \ar@{|-}[rr]_-*--{\scriptscriptstyle #1}^-*--{\scriptstyle \ast}
  &*{\phantom{\scriptscriptstyle{#1}}}&}
}{}}}}
\newcommand{\partintostarnosmash}[1][]{\mathord{\mathchoice{%
  \xymatrix@=0.4em@1{%
  \ar@{|-}[rr]_-*--{\scriptstyle #1}^-*--{\scriptstyle \ast}
  &*{\phantom{\scriptstyle{#1}}}&}
}{
  \xymatrix@=0.25em@1{%
  \ar@{|-}[rr]_-*--{\scriptstyle #1}^-*--{\scriptstyle \ast}
  &*{\phantom{\scriptstyle{#1}}}&}
}{
  \xymatrix@=0.2em@1{%
  \ar@{|-}[rr]_-*--{\scriptscriptstyle #1}^-*--{\scriptstyle \ast}
  &*{\phantom{\scriptscriptstyle{#1}}}&}
}{}}}

\newcommand{\sk}{\textup{sk}}
\newcommand{\ak}{\textup{ak}}
\newcommand{\am}{\textup{am}}
\newcommand{\sm}{\textup{sm}}

\theoremstyle{definition}

\title{Geometric complexity theory and matrix powering}
\author{Fulvio Gesmundo$^1$, Christian Ikenmeyer$^2$, Greta Panova$^3$}

\date{\today}

\begin{document}
\sloppy

\begin{abstract}
Valiant's famous determinant versus permanent problem is the flagship problem in algebraic complexity theory.
Mulmuley and Sohoni (Siam J Comput 2001, 2008)
introduced geometric complexity theory, an approach to study this and related problems via algebraic geometry and representation theory.
Their approach works by multiplying the permanent polynomial with a high power of a linear form (a process called padding)
and then comparing the orbit closures of the determinant and the padded permanent.
This padding was recently used heavily to show negative results
for the method of shifted partial derivatives (Efremenko, Landsberg, Schenck, Weyman, 2016)
and for geometric complexity theory (Ikenmeyer Panova, FOCS 2016 and B\"urgisser, Ikenmeyer Panova, FOCS 2016),
in which occurrence obstructions were ruled out to be able to prove superpolynomial complexity lower bounds.
Following a classical homogenization result of Nisan (STOC 1991) we replace the determinant in geometric complexity theory with the trace of a symbolic matrix power.
This gives an equivalent but much cleaner homogeneous formulation of geometric complexity theory in which the padding is removed.
This radically changes the representation theoretic questions involved to prove complexity lower bounds.
We prove that in this homogeneous formulation there are no orbit occurrence obstructions that prove even superlinear
lower bounds on the complexity of the permanent.

Interestingly---in contrast to the determinant---the trace of a symbolic matrix power is not uniquely determined by its stabilizer.
\end{abstract}
\maketitle

\footnotetext[1]{QMATH, Dept. of Math. Sciences, Univ. of Copenhagen, Denmark, (Texas A\&M University, College Station, TX) \\ \phantom{a} \hfill \email{fulvio.gesmundo@gmail.com}}
\footnotetext[2]{Max Planck Institute for Informatics, Saarland Informatics Campus, Germany, \email{cikenmey@mpi-inf.mpg.de}}
\footnotetext[3]{University of Pennsylvania, Philadelphia, PA, \email{panova@math.upenn.edu}}

\section{Statement of the result}
Let $\per_m := \sum_{\sigma \in \aS_m} \prod_{i=1}^m X_{i,\sigma(i)}$ denote the $m \times m$ permanent polynomial
and let $\Pow_n^m := \tr(X^m)$ denote the trace of the $m$th power of an $n \times n$ matrix $X=(X_{i,j})$ of variables.
The coordinate rings of the orbits and orbit closures $\IC[\GL_{n^2}\Pow^m_n]$ and $\IC[\overline{\GL_{n^2}\per_m}]$
are $\GL_{n^2}$-representations.
Let $\la$ be an isomorphism type of an irreducible $\GL_{n^2}$-representation.
In this paper we prove that if $n \geq m+2 \geq 12$ and
$\la$ occurs in $\IC[\overline{\GL_{n^2}\per_m}]$,
then $\la$ also occurs in $\IC[\GL_{n^2}\Pow^m_n]$, see Theorem~\ref{thm:main} below.

\section{Introduction}\label{section: introduction}
Valiant's famous determinant versus permanent problem is a major open problem in computational complexity theory.
It can be stated as follows, see Conjecture~\ref{conj:valiantconj}:
For a polynomial $p$ in any number of variables
let the determinantal complexity $\dc(p)$ denote the smallest $n \in \IN$ such that $p$
can be written as the determinant $p=\det(A)$ of an $n \times n$ matrix $A$ whose entries are affine linear forms in the variables.

Throughout the paper we fix our ground field to be the complex numbers $\IC$.
The permanent is of interest in combinatorics and theoretical physics, but our main interest
stems from the fact that it is complete for the complexity class $\VNP$ (although the arguments in this paper remain valid
if the permanent is replaced by any other $\VNP$-complete function, mutatis mutandis).
Valiant famously posed the following conjecture.
\begin{conjecture}\label{conj:valiantconj}
The sequence $\dc(\per_m)$ grows superpolynomially.
\end{conjecture}
Valiant \cite{Val:Complete_classes_in_algebra} proved that Conj.~\ref{conj:valiantconj} implies the separation $\VPe \subsetneq \VNP$ of algebraic complexity classes,
which was later refined in \cite{Toda:Classes_of_arithm_circuits_for_det}, see also \cite{MaPo:Characterizing_Valiant_alg_compl_classes}:
Conj.~\ref{conj:valiantconj} is equivalent to the separation $\VPs \subsetneq \VNP$.
Many polynomially equivalent formulations for the determinantal complexity exist. For example
$\dc(p)$ is polynomially equivalent to the smallest size of a skew circuit computing $p$,
or the smallest size of a weakly skew circuit computing $p$,
or the smallest size of an algebraic branching program computing $p$.

\subsection{Preliminaries and the padded setting}
Geometric complexity theory was introduced by Mulmuley and Sohoni \cite{GCT1,GCT2} to resolve Conj.~\ref{conj:valiantconj} and related conjectures as follows.
For $n > m$ define the \emph{padded permanent} $\per_m^n := (X_{n,n})^{n-m}\per_m$, which is homogeneous of degree $n$ in $m^2+1$ variables.
Let $\IA_n$ denote the vector space of homogeneous degree $n$ polynomials in the $n^2$ variables $X_{i,j}$.
Clearly $\per_m^n \in \IA_n$.
Moreover, $\det_n \in \IA_n$, where $\det_n := \sum_{\sigma \in \aS_n} \sgn(\sigma) \prod_{i=1}^n X_{i,\sigma(i)}$ is the determinant polynomial.
The group $\GL_{n^2}$ of invertible $n^2 \times n^2$ matrices acts canonically on $\IA_n$ by replacing variables with homogeneous linear forms.
Let $\GL_{n^2} \det_n := \{g \cdot \det_n \mid g \in \GL_{n^2}\}$ be the orbit of the determinant
and analogously let $\GL_{n^2} \per_m^n$ be the orbit of the padded permanent.
Let $\overline{\GL_{n^2}\det_n} \subseteq \IA_n$ and $\overline{\GL_{n^2}\per_m^n} \subseteq \IA_n$ denote the closures of the respective orbits in $\IA_n$.
Here Euclidean closure and Zariski closure coincide \cite[II.2.2 c \& AI.7.2 Folgerung]{Kra:Geom_Meth_Invarianterntheorie}, i.e.,
both orbit closures are affine subvarieties of $\IA_n$.
Mulmuley and Sohoni proposed the following way to find lower bounds on $\dc(\per_m)$.
\begin{proposition}\label{pro:monoids}
If $\overline{\GL_{n^2}\per_m^n} \not\subseteq \overline{\GL_{n^2}\det_n}$, then $\dc(\per_m)>n$.
\end{proposition}
We call Prop.~\ref{pro:monoids} the \emph{padded setting}.
To prove lower bounds on $\dc(\per_m)$ Mulmuley and Sohoni \cite{GCT1,GCT2} suggested to study the representation theory of the coordinate rings of
the orbits and orbit closures and use so-called \emph{occurrence obstructions}.
To define occurrence obstructions we now discuss the representation theory of the coordinate rings.

Recall that $\IA_n$ is a complex vector space of dimension $\binom{n^2+n-1}{n}$.
Let $\IC[\IA_n]$ denote its coordinate ring, i.e., the ring of polynomials in $\binom{n^2+n-1}{n}$ variables.
The group $\GL_{n^2}$ acts linearly in a canonical way on each homogeneous degree $d$ component $\IC[\IA_n]_d$ by the canonical pullback $(gf)(p) := f(g^{-1}p)$,
for all $f \in \IC[\IA_n]_d$, $g \in \GL_{n^2}$, $p \in \IA_n$.
Since the group $\GL_{n^2}$ is reductive, the finite dimensional $\GL_{n^2}$-representation $\IC[\IA_n]_d$ splits into a direct sum
representations:
$
\IC[\IA_n]_d = \bigoplus_{i} V_i,
$
where each $V_i$ is an irreducible $\GL_{n^2}$-representation, i.e., a vector space with no nontrivial linear subspaces that are invariant under the group action.
Two irreducible $\GL_{n^2}$-representations $V_i$ and $V_j$ are called \emph{isomorphic}
if there exists a $\GL_{n^2}$-equivariant vector space isomorphism $\varphi : V_i \stackrel{\sim}{\to} V_j$,
i.e. $g\varphi(v) = \varphi(gv)$ for all $g \in \GL_{n^2}$ and $v \in V_i$.
Grouping together isomorphic copies we write
$
\IC[\IA_n]_d = \bigoplus_{\la} V_\la^{\oplus a_\la(d[n])},
$
where the natural numbers $a_\la(d[n])$ are the so-called \emph{plethysm coefficients} and the sum ranges over all isomorphism types $\la$ of irreducible $\GL_{n^2}$-representations.
It is a major open problem in algebraic combinatorics to find a combinatorial description for $a_\la(d[n])$, see Problem 9 in \cite{Sta:Positivity_probs_conj_alg_comb}.

A \emph{partition} is a finite sequence of nonincreasing nonnegative integers.
Partitions are often depicted by their Young diagrams, which are top-left justified arrays of boxes, where there are $\la_i$ boxes in row $i$.
For example, the Young diagram of the partition $(5,4)$ is $\tiny\yng(5,4)$.
We often identify partitions with their Young diagrams.
We write $|\la| := \la_1 + \la_2 + \cdots$ for the number of boxes in $\la$. Moreover, we write $\la \partinto[N]D$
to denote a partition $\la$ of $D$ into at most $N$ parts.
We omit $N$ if there is no restriction on the number of parts.
We denote by $\ell(\la)$ the \emph{length} of $\la$, which is its number of nonzero parts.
The isomorphism types or irreducible $\GL_{n^2}$-representations that can possibly appear in $\IC[\IA_n]_d$
are indexed by partitions $\la \partinto[n^2]nd$.

For an affine subvariety $Z \subseteq \IA_n$ (e.g., $Z=\overline{\GL_{n^2}\det_n}$ or $Z=\overline{\GL_{n^2}\per_m^n}$)
the coordinate ring $\IC[Z]$ is defined by restricting the functions in $\IC[\IA_n]$ to $Z$.
Since in our case $Z$ will always be a cone (i.e., closed under vector space rescaling) it follows that
$\IC[Z]$ inherits the grading from $\IC[\IA_n]$. In each degree $d$ both coordinate rings split:
$
\IC[\overline{\GL_{n^2}\det_n}]_d = \bigoplus_{\la} V_\la^{\oplus a'_\la(d[n])} $ and $
\IC[\overline{\GL_{n^2}\per_m^n}]_d = \bigoplus_{\la} V_\la^{\oplus a''_{\la,m}(d[n])},
$
where the $a'_\la(d[n])$ and $a''_\la(d[n])$ are nonnegative integers.
Clearly if $\overline{\GL_{n^2}\per_m^n} \subseteq \overline{\GL_{n^2}\det_n}$, then there exists the
surjective $\GL_{n^2}$-equivariant homomorphism $\IC[\overline{\GL_{n^2}\det_n}]_d \twoheadrightarrow \IC[\overline{\GL_{n^2}\per_m^n}]_d$
that is just the restriction of functions.
Schur's lemma implies that in this case we have $a''_{\la,m}(d[n]) \leq a'_\la(d[n])$.
Using Prop.~\ref{pro:monoids} we can draw the following conclusion.
\begin{proposition}\label{pro:reprthobs}
If there exists a partition $\la$ such that $a'_\la(d[n]) < a''_{\la,m}(d[n])$, then $\dc(\per_m)>n$.
\end{proposition}
These $\la$ are called \emph{representation theoretic obstructions}.
Mulmuley and Sohoni conjectured that Prop.~\ref{pro:reprthobs} can be used to resolve Conj.~\ref{conj:valiantconj}.
This conjecture is still wide open.
A key insight about the shape of possible $\lambda$ is presented in \cite{KadLan:padded_poly_and_GCT}:
\begin{proposition}\label{pro:kadishlandsberg}
If $a''_{\la,m}(d[n])>0$, then $\la_1\geq (n-m)d$.
\end{proposition}
This poses a crucial restriction to the possible obstructions $\la$:
If we search for obstructions that prove Conj.~\ref{conj:valiantconj},
then we can assume that the first part of $\la$ is much larger than its other parts.

Mulmuley and Sohoni proposed to search for $\la$ that satisfy not only $a'_\la(d[n]) < a''_{\la,m}(d[n])$ but even $a'_\la(d[n]) = 0 < a''_{\la,m}(d[n])$.
Such $\lambda$ are called \emph{occurrence obstructions}.
It was recently shown that  no lower bounds better than $\dc(\per_m)>m^{25}$ can be proved with occurrence obstructions \cite{BuIkPa:no_occurrence_obstructions_in_GCT}.
Mulmuley and Sohoni proposed even further to use the following upper bound for $a'_\la(d[n])$ coming from the \emph{coordinate ring of the determinant orbit}:
The algebraic group $\GL_{n^2}$ is an affine variety and acts on itself by left and right multiplication.
Hence $\GL_{n^2} \times \GL_{n^2}$ acts on the coordinate ring $\IC[\GL_{n^2}]$.
The algebraic Peter-Weyl theorem (see e.g.\ \cite[II.3.1 Satz 3]{Kra:Geom_Meth_Invarianterntheorie}, \cite[Ch.~7, 3.1 Thm.]{Pro:LieGroups}, or \cite[Thm.~4.2.7]{GooWal:Symmetry_reps_invs}) tells us how its coordinate ring splits as a $\GL_{n^2}\times \GL_{n^2}$-representation:
$
\IC[\GL_{n^2}] \simeq \bigoplus_{\la} V_\la \otimes V_{\la^*},
$
where the sum is over all isomorphism types of $\GL_{n^2}$ and $\la^*$ is the type dual to $\la$.
If $p \in \IA_n$ has a reductive stabilizer $\calS\subseteq \GL_{n^2}$,
then the orbit $\GL_{n^2}p$ is an affine variety whose coordinate ring $\IC[\GL_{n^2}p]$ is the ring of right $\calS$-invariants:
$\IC[\GL_{n^2}p] = \IC[\GL_{n^2}]^\calS$, see \cite[Sec.~4.1 \& 4.2]{BuIkPa:no_occurrence_obstructions_in_GCT}.
For the determinant the stabilizer was already calculated by Frobenius \cite{Frob:Stabilizer_of_det}.
Functions on the orbit closure restrict to the orbit and since the orbit is dense in its closure this gives an embedding $\IC[\overline{\GL_{n^2}p}] \subseteq \IC[\GL_{n^2}p]$ and in each degree $d$ we have that
\begin{equation}
\label{eq:subrep}
\IC[\overline{\GL_{n^2}p}]_d \subseteq \IC[\GL_{n^2}p]_d
\end{equation}
is a $\GL_{n^2}$-subrepresentation; see also \cite{BI:15} for a study of the relationship between the two coordinate rings.
The multiplicities that arise in $\IC[\GL_{n^2}\det_n]_d$ are much more accessible than those in $\IC[\overline{\GL_{n^2}\det_n}]_d$. Indeed,
$
\IC[\GL_{n^2}\det_n]_d = \bigoplus_{\la} V_\la^{\oplus \sk(\la,n\times d)}.
$
Here the nonnegative integer $\sk(\la,n\times d)$ is the so-called \emph{rectangular symmetric Kronecker coefficient},
a quantity that can be described completely in terms of the symmetric group as follows.
The irreducible representations of the symmetric group $\aS_{D}$ are indexed by partitions $\la$ of $D$ into arbitrarily many parts and denoted by $[\la]$.
For partitions $\la \partinto D$ and $\mu \partinto D$ the group $\aS_D \times \aS_D$ acts irreducibly on the tensor product $[\la] \otimes [\mu]$,
but the embedding $\aS_D \hookrightarrow \aS_D \times \aS_D$, $g \mapsto (g,g)$ makes $[\la] \otimes [\mu]$ an $\aS_D$ representation that decomposes:
$
[\la] \otimes [\mu] \simeq \bigoplus_\nu [\nu]^{\oplus g(\la,\mu,\nu)},
$
where the sum is over all partitions of $D$
and the nonnegative integers $g(\la,\mu,\nu)$ are called the \emph{Kronecker coefficients}.
Finding a combinatorial expression for $g(\la,\mu,\nu)$ is a famous open problem in algebraic combinatorics (see Problem 10 in \cite{Sta:Positivity_probs_conj_alg_comb}).
If we replace the tensor product $[\la]\otimes[\la]$ by the symmetric square $S^2[\la]$ (which is the space of $\IZ_2$ invariants in $[\la]\otimes[\la]$,
where $\IZ_2$ switches the tensor factors),
we get the \emph{symmetric Kronecker coefficients}:
$
S^2[\la] \simeq \bigoplus_\mu [\mu]^{\oplus \sk(\mu,\la)},
$
where the sum is over all partitions of $D$
and the nonnegative integers $\sk(\mu,\la)$ are called the \emph{symmetric Kronecker coefficients}.
We denote by $n \times d := (d,d,\ldots,d)$ the partition of $nd = D$ in which each of the $n$ parts equals $d$ and call it the \emph{rectangular partition}
because its Young diagram is a rectangle.
Using this notation $\sk(\la,n \times d)$ is the multiplicity of $[\la]$ in $S^2[n \times d]$.

From \eqref{eq:subrep} it follows that $a'_\la(d[n]) \leq \sk(\la,n\times d)$. A partition $\la$ that satisfies $\sk(\la,n\times d)=0<a''_{\la,m}(d[n])$ is called an \emph{orbit occurrence obstruction}.
Its existence implies $\dc(\per_m)>n$.
Mulmuley and Sohoni conjectured that orbit occurrence obstructions which prove Conj.~\ref{conj:valiantconj} exist,
recently disproved in \cite{BuIkPa:no_occurrence_obstructions_in_GCT}.
A natural upper bound for $\sk(\la,n\times d)$ is the Kronecker coefficient $g(\la,n\times d,n\times d)$.
Mulmuley and Sohoni conjectured that the vanishing of $g(\la,n\times d,n\times d)$ suffices to
find sufficiently good orbit occurrence obstructions that prove Conj.~\ref{conj:valiantconj},
but recently \cite{IkPa:Rectangular_Kron_in_GCT} proved that no lower bounds better than $3m^4$ can be proved in this way.
Note that even the small polynomial $3m^4$ would be a highly nontrivial lower bound:
The best lower bound on $\dc(\per_m)$ is $\frac {m^2} 2$ by Mignon and Ressayre \cite{MigRes:Quadratic_bound_on_dc_perm}.
The paper \cite{IkPa:Rectangular_Kron_in_GCT} does not rule out that this lower bound could be improved using orbit occurrence obstructions
and their proof is tightly optimized to yield an exponent as small as possible.
Notably \cite{IkPa:Rectangular_Kron_in_GCT} does not make a statement about symmetric Kronecker coefficients because they are more challenging than Kronecker coefficients.
We will see in Section~\ref{sec:skcof} how trivial statements about Kronecker coefficients can become interesting if one studies symmetric Kronecker coefficients.

What all these different coefficients have in common are the \emph{semigroup properties}, which are all proved in the same way by multiplying two highest weight vectors.
Let $\la=(\la_1,\la_2,\ldots)$ and $\mu=(\mu_1,\mu_2,\ldots)$ be partitions, then $\la+\mu$ is defined as $(\la_1+\mu_1,\la_2+\mu_2,\ldots)$.
The semigroup property states that $a'_{\la+\mu}((d+d')[n]) \geq \max\{a'_\la(d[n]),a'_\mu(d'[n])\}$ and $g(\la+\alpha,\mu+\beta,\nu+\gamma) \geq \max\{g(\la,\mu,\nu),g(\alpha,\beta,\gamma)\}$.
Analogous properties hold for many other coefficients, e.g., for $a_\la(d[n])$, $a''_{\la,m}(d[n])$, and $\sk(\la, n \times d)$.

The results in \cite{IkPa:Rectangular_Kron_in_GCT} ($g(\la,n\times d,n\times d)>0$) and \cite{BuIkPa:no_occurrence_obstructions_in_GCT} ($a'_\la(d[n])>0$) are proved using the semigroup property in the following way:
One decomposes $\la$ into a sum of smaller partitions, then shows positivity for the smaller partitions, and then uses the semigroup property.
In both papers Prop.~\ref{pro:kadishlandsberg} is heavily used
because it enables us to assume that the smaller partitions have an almost arbitrarily chosen first part.
This simplifies the construction of these positive building blocks considerably.
Prop.~\ref{pro:kadishlandsberg} crucially uses that the permanent is padded with a high power of a linear form.

Moreover, also crucially using this padding,
\cite{EfrLanSchWey:shifted_partials_and_perm_vs_det} showed that the method of shifted partial derivatives applied to Prop.~\ref{pro:monoids} cannot be used to prove Conj.~\ref{conj:valiantconj}.

In the light of these no-go results we remove the necessity of the padding in the next section.

\subsection{The homogeneous setting}
Using a result by Nisan \cite{Nisan:Lower_bounds_non_commutative_compts} Prop.~\ref{pro:monoids} and the whole geometric complexity theory approach can be reformulated \emph{without padding the permanent}:
Let $\IA_n^m$ denote the space of homogeneous degree $m$ polynomials in $n^2$ variables.
Let
$
\Pow_n^m := \tr(X^m) \in \IA_n^m, \text{ where } X=(X_{i,j}) \text{ is the $n \times n$ variable matrix}.
$
Let $\pc(\per_m)$ denote the smallest $n$ such that
$\per_m$ can be written as $p=\tr(A^m)$, where $A$ is an $n \times n$ matrix whose entries are \emph{homogeneous} linear forms.

It is well known (\cite[Lem.~1]{Nisan:Lower_bounds_non_commutative_compts}, see also~\cite[Exe.~5.1]{Sap:LowerBounds_in_arithm_circuits_compl} or \cite[Rem.~4.5]{IkLan:Compl_of_perm_in_various_comp_models})
that $\pc(\per_m)$ and $\dc(\per_m)$ are polynomially equivalent and hence Conj.~\ref{conj:valiantconj} is equivalent to
\begin{equation}\label{eq:valiantconjII}
\text{the sequence $\pc(\per_m)$ grows superpolynomially}.
\end{equation}
Interestingly, the proof of the best known upper bound $\dc(\per_m)\leq 2^m-1$ by Grenet \cite{Gre:11} also works for this measure: $\pc(\per_m)\leq 2^m-1$. Completely analogously to Prop.~\ref{pro:monoids} one can show
\begin{proposition}\label{pro:monoidsII}
If $\overline{\GL_{n^2}\per_m} \not\subseteq \overline{\GL_{n^2}\Pow_n^m}$, then $\pc(\per_m)>n$.
\end{proposition}
We call Prop.~\ref{pro:monoidsII} the \emph{homogeneous setting}.
To study representation theoretic obstructions in the homogeneous setting we consider the splitting of the coordinate rings in the same way as in the padded setting:
$
\IC[\overline{\GL_{n^2}\per_m}]_d = \bigoplus_\la V_\la^{\oplus q_\la(d[m])}$ and $\IC[\overline{\GL_{n^2}\Pow_n^m}]_d = \bigoplus_\la V_\la^{\oplus t_{\la,n}(d[m])}.
$
It follows from \cite[Thm.~6.1.5]{BuLaMaWe:Ov_math_iss_in_GCT} that $q_\la(d[m])$ does not depend on $n$ for $n \geq m$, so the notation is justified.
As in the padded setting Schur's lemma implies:
\begin{corollary}
If $q_\la(d[m])>t_{\la,n}(d[m])$, then $\pc(\per_m)>n$.
\end{corollary}
In Section~\ref{sec:stabinv_semigroup} we calculate how the coordinate ring of the orbit $\GL_{n^2}\Pow_n^m$ splits.
This is based on knowing the stabilizer of $\Pow_n^m$:
\begin{theorem}\label{thm:stabilizer}
Let $X=(X_{i,j})$ be an $n \times n$ variable matrix. Then
$\tr(X^m) = \tr((X^t)^m)$ and $\tr(X^m) = \tr((gXg^{-1})^m)$, where $g \in \GL_n$, and $\tr(X^m) = \tr((\omega X)^m)$, where $\omega$ is an $m$-th root of unity. Moreover, if $n,m \geq 3$, the whole stabilizer $\calS$ of $\Pow_n^m$ is generated by these symmetries.
\end{theorem}
Theorem~\ref{thm:stabilizer} is proved in Section~\ref{sec:stab}.
\begin{theorem}\label{thm:sm} For $n,m\geq 3$ we have
$
\IC[\GL_{n^2}\Pow_n^m]_d = \IC[\GL_{n^2}]^{\calS}_d = \bigoplus_\la V_\la^{\oplus \sm(\la,n)},
$
where the sum is over all $\la \partinto md$ and $\sm(\la,n) := \sum_{\mu \partinto[n]dm} \sk(\la,\mu)$ is a sum of symmetric Kronecker coefficients.
\end{theorem}
Note that $\sm(\la,n)$ does not depend on $d$ and $m$ independently, but only on their product $dm = |\la|$, therefore the notation $\sm(\la,n)$ is justified.

Analogously to the padded setting we have the inclusion $\IC[\overline{\GL_{n^2}\Pow_n^m}]_d \subseteq \IC[\GL_{n^2}\Pow_n^m]_d$ of $\GL_n^2$-representations,
therefore $t_{\la,n}(d[m]) \leq \sm(\la,n)$.

\begin{corollary}\label{cor:hope}
If $\la\partinto dm$ and $\sm(\la,n)=0<q_\la(d[m])$, then $\pc(\per_m)>n$.
\end{corollary}
We call these $\la$ \emph{orbit occurrence obstructions}.
We prove that no super\emph{linear} lower bounds can be proved with orbit occurrence obstructions:
\begin{theorem}[{\bf Main Result}]\label{thm:main}
Let $m \geq 10$ and $n \geq m+2$.
For every $\la\partinto dm$ that satisfies $q_\la(d[m])>0$ we have $\sm(\la,n)>0$.
\end{theorem}

This is the first time that the possibility of superlinear lower bounds is ruled out in geometric complexity theory.

Note that in contrast to \cite{IkPa:Rectangular_Kron_in_GCT} we work directly with the multiplicities in the coordinate ring of the orbit and not with any upper bound.

The methods used to prove this result differ greatly from \cite{IkPa:Rectangular_Kron_in_GCT},
in particular \cite{BuIkPa:no_occurrence_obstructions_in_GCT} lifts the result in \cite{IkPa:Rectangular_Kron_in_GCT} to the closure, which appears to be challenging in the homogeneous setting because of the absence of the padding.

\begin{remark}
We remark that even though the homogeneous setting is equivalent to the padded setting in terms of algebraic complexity theory in a very natural way,
$\Pow_n^m$ is \emph{not} characterized by its stabilizer (see Prop.~\ref{prop: space invariant under S}), unlike the determinant.
Obtaining a homogeneous setting in which the computational model is characterized by its stabilizer is also possible:
one has to study the orbit closure of the $m$-factor iterated $n \times n$ matrix multiplication,
a polynomial in $m n^2$ variables,
which seems to be even more challenging.
Its stabilizer has been identified in \cite{Ges:Geometry_of_IMM}.
\end{remark}

\begin{proof}[Proof of Theorem~\ref{thm:main}]
We start with some simple observations on $q_\la(d[m])$.
\begin{lemma}\label{lem:easyrestronla}
$q_\la(d[m]) \leq a_\la(d[m])$.
Moreover, if $q_\la(d[m])>0$, then $\ell(\la)\leq m^2$ and $|\la|=md$.
\end{lemma}
\begin{proof}
The coordinate ring $\IC[\IA_n^m]$ splits according to the plethysm coefficients
$
\IC[\IA_n^m]_d = \bigoplus_{\la \partintonosmash[n^2]md} V_\la^{\oplus a_\la(d[m])}.
$
Therefore $q_\la(d[m]) \leq a_\la(d[m])$, because the orbit closure $\overline{\GL_{n^2}\per_m}$ is an affine subvariety of $\IA_n^m$.
See e.g.~\cite[Lem.~4.3.3]{Ike:PhDthesis} for the classical $|\la|=md$.
Lastly, $\ell(\la)\leq m^2$ is ensured by \cite{BuLaMaWe:Ov_math_iss_in_GCT}, just because the $\per_m$ has only $m^2$ variables.
\end{proof}

In order to prove Theorem~\ref{thm:main} we assume that $q_\la(d[m])>0$ for some partition $\la$.
By Lemma~\ref{lem:easyrestronla} this implies $|\la|=dm$, $\ell(\la)\leq m^2$, and $a_\la(d[m])>0$.
The following Prop.~\ref{pro:plethvanish} ensures that $\la_1 \geq m$.
\begin{proposition}\label{pro:plethvanish}
If $\la_1 < m$, then $a_\la(d[m])=0$.
\end{proposition}
We prove Prop.~\ref{pro:plethvanish} in Section~\ref{sec:vanishingpleth}.
Since Lemma~\ref{lem:easyrestronla} implies $\ell(\la)\leq m^2$,
we can conclude the proof of Thm.~\ref{thm:main} with the following positivity proposition.

\begin{proposition}\label{pro:pos}
Let $m \geq 10$, $d \geq 1$. Further let $\la \partinto md$, $\la_1 \geq 3$, $\ell(\la)\leq m^2$.
If $n \geq m+2$, then $\sm(\la,n)>0$.
\end{proposition}
We prove a slightly more general result in Section~\ref{sec:smpositivity} (Proposition~\ref{prop:sm_positive}, where $L=m^2$).
\end{proof}

\section*{Acknowledgments}
We thank Neeraj Kayal, Michael Forbes, and Pierre Lairez for helpful discussions on arithmetic circuits.
We thank Peter B\"urgisser for helpful insights on the subgroup restriction problem at hand. The third author was partially supported by NSF.

\section{Occurrence in the coordinate ring of the orbit}\label{sec:smpositivity}
Here we prove that the relevant multiplicities $\sm(\la,n)$ are positive in all cases of interest, and in particular we prove Proposition~\ref{pro:pos}.
We list the necessary facts, their proofs appear in the corresponding sections.

Analogously to $\sk(\lambda,\mu)=\dim ([\lambda]\otimes S^2[\mu])^{\aS_D}$ let 
$\ak(\lambda,\mu)=\dim ([\lambda]\otimes \Lambda^2[\mu])^{\aS_D}$ denote the multiplicity of $[\lambda]$ in $\Lambda^2[\mu]$.
Since $S^2[\mu] \oplus \Lambda^2[\mu]= [\mu] \otimes [\mu]$, we trivially have 
\begin{align}\label{eq:g=sk+ak} 
g(\lambda,\mu,\mu) = \sk(\lambda,\mu)+\ak(\lambda,\mu).
\end{align}
Moreover, for any positive integer $a$ and any partition $\lambda$ we set
$$\sm(\lambda,a) := \sum_{\mu: \ell(\mu) \leq a} \sk(\lambda,\mu) \quad \text{ and } \quad \am(\lambda,a) := \sum_{\mu: \ell(\mu) \leq a} \ak(\lambda,\mu).$$

A crucial property to prove positivity is the semigroup property. Informally it follows from multiplying highest weight vectors in invariant spaces.

\begin{proposition}[Semigroup properties]\label{prop:semigroup}
Let $\la$ and $\nu$ be partitions. We have
\begin{compactenum}[(1)]
\item If $\sm(\la,n)>0$ and $\sm(\nu,n)>0$, then $\sm(\la+\nu,n) \geq \max(\sm(\la,n),\sm(\nu,n))$.
\item If $\am(\la,n)>0$ and $\am(\nu,n)>0$, then $\sm(\la+\nu,n) \geq \max(\am(\la,n),\am(\nu,n))$.
\item If $\sm(\la,n)>0$ and $\am(\nu,n)>0$, then $\am(\la+\nu,n) \geq \max(\sm(\la,n),\am(\nu,n))$.
\end{compactenum}
\end{proposition}
If $|\la|$ is a multiple of some $m \geq 3$, i.e., $|\la|=dm$, then by Theorem~\ref{thm:sm} $\sm(\la,n)$ is the multiplicity of $\la$ in
$\IC[\GL_{n^2}\Pow_n^m]_d$ and thus part (1) holds by multiplying two highest weight vector functions,
provided both $|\la|$ and $|\nu|$ are divisible by the same number $m \geq 3$.
The approach to proving the general case is very similar, as we will see next.
\begin{proof}[Proof of Proposition~\ref{prop:semigroup}]
We write $\la\partinto[n] d$ to denote that $\la$ is a partition of $d$ into at most $n$ parts.
Let $V \simeq \IC^n$ and let $V^*$ be its dual.
The space $V^* \otimes V$ is naturally isomorphic to $V^* \otimes V^{**} = V^* \otimes V$.
This gives rise to a natural automorphism on $V^* \otimes V$ that has order 2,
i.e., this gives an $\aS_2$ action on $V^* \otimes V$.
Thus we get an $\aS_2$ action on $V \otimes V^* \otimes V$ keeping the first tensor factor fixed.
This induces and $\aS_2$ action on $\tensor^d (V \otimes V^* \otimes V)$.
Since the $\aS_2$ action commutes with the natural action of $\aS_d$ on $\tensor^d (V \otimes V^* \otimes V)$,
we have an $\aS_2$ action on the $\aS_d$ invariant space $\Sym^d (V \otimes V^* \otimes V)$.
Let $\GL := \GL(V)$ and $G := \GL^2$.
Embed $G \hookrightarrow \GL^3$ via $(g,h) \mapsto (g,h,h)$.
In this way $G$ acts of $\Sym^d (V \otimes V^* \otimes V)$ and 
the actions of $\aS_2$ and $G$ commute.
Thus we have an action of $\aS_2$ on the highest weight vector space
$\HWV_{\la,\mu}(\Sym^d \tensor^3 V)$.

Schur-Weyl duality says
\[
\tensor^d V \simeq \bigoplus_{\la\partinto[n] d} [\la] \otimes \{\la\}.
\]
Thus we have
\begin{eqnarray*}
\tensor^d (V \otimes V^* \otimes V) \simeq \bigoplus_{\la,\mu,\nu \partinto[n] d} [\la] \otimes \{\la\} \otimes [\mu] \otimes \{\mu^*\} \otimes [\nu] \otimes \{\nu\},
\end{eqnarray*}
as $\GL^3$-modules, where $\mu^*=(-\mu_n,-\mu_{n-1},\ldots,-\mu_1)$, so that $\{\mu^*\}$ is the $\GL$-module dual to $\{\mu\}$.
We write $\rho \vDash_n k$ for a nonincreasing sequence of $n$ integers that sum up to $k$.
As $G$-modules we have
\begin{eqnarray*}
\tensor^d (V \otimes V^* \otimes V) \simeq \bigoplus_{\la,\mu,\nu\vdash d, \ \rho\vDash 0} c_{\mu,\nu^*}^\rho [\la] \otimes [\mu] \otimes [\nu] \otimes \{\la\} \otimes \{\rho\},
\end{eqnarray*}
where $c_{\mu,\nu^*}^\rho$ is the Littlewood-Richardson coefficient (which is naturally defined not only for partitions, but for nonincreasing sequences of integers).
Recall $G = \GL\times\GL$. We want to distinguish between the left and the right factor and therefore we denote by $H$ the right factor, i.e. $G = \GL\times H$.
Going to $H$-invariants we see that the Littlewood-Richardson coefficients are either 1 or 0, so
\begin{eqnarray*}
\tensor^d (V \otimes V^* \otimes V)^{H} \simeq \bigoplus_{\la,\mu\partinto[n] d} [\la] \otimes [\mu] \otimes [\mu] \otimes \{\la\}.
\end{eqnarray*}

Since $\{\la\}$ contains a unique highest weight vector line of type $\la$ and no other highest weight vector,
going to $\GL$-highest weight vector spaces yields
\begin{eqnarray*}
\HWV_{\la} \Big(\tensor^d (V \otimes V^* \otimes V)^{H} \Big) \simeq \bigoplus_{\la,\mu\partinto[n] d} [\la] \otimes [\mu] \otimes [\mu],
\end{eqnarray*}
because there is a unique HWV in every $\GL$-representation.

$\aS_2$ acts on this space and we take invariants:

\begin{eqnarray*}
\HWV_{\la} \Big(\tensor^d (V \otimes V^* \otimes V)^{H} \Big)^{\aS_2} \simeq \bigoplus_{\la,\mu\partinto[n] d} [\la] \otimes \Sym^2([\mu]).
\end{eqnarray*}

Since the action of $\aS_d$ commutes with the actions of $\GL^3$ and $\aS_2$, we can take $\aS_d$ invariants and obtain
\begin{eqnarray*}
\HWV_{\la} \Big(\Sym^d (V \otimes V^* \otimes V)^{H} \Big)^{\aS_2} \simeq \bigoplus_{\la,\mu\partinto[n] d} \IC^{\oplus \sk(\la,\mu)}.
\end{eqnarray*}

Completely analogously we can take $\aS_d$ skew-invariants (denoted by $\textup{skew-}\aS_2$) and obtain
\begin{eqnarray*}
\HWV_{\la} \Big(\Sym^d (V \otimes V^* \otimes V)^{H} \Big)^{\textup{skew-}\aS_2} \simeq \bigoplus_{\la,\mu\partinto[n] d} \IC^{\oplus \ak(\la,\mu)}.
\end{eqnarray*}

We conclude the proof by analyzing what happens when we multiply two
highest weight vector polynomials, 
$f \in \HWV_{\la} \Big(\Sym^d (V \otimes V^* \otimes V)^{H} \Big)$
and
$g \in \HWV_{\nu} \Big(\Sym^d (V \otimes V^* \otimes V)^{H} \Big)$.
The product $fg$ is a highest weight vector in
$\HWV_{\la+\nu} \Big(\Sym^d (V \otimes V^* \otimes V)^{H} \Big)$.
If $f$ and $g$ are both $\aS_2$-invariant, then their product $fg$ is $\aS_2$-invariant.
But if $f$ and $g$ are both $\aS_2$-skew-invariant, then $fg$ is also $\aS_2$-invariant.
Moreover, if $f$ is $\aS_2$-invariant and $g$ is $\aS_2$-skew-invariant, then $fg$ is $\aS_2$-skew-invariant.
The inequality that we need to show follows from multiplying a basis of
$\HWV_{\la} \Big(\Sym^d (V \otimes V^* \otimes V)^{H} \Big)^{\aS_d}$
with a single polynomial in
$\HWV_{\nu} \Big(\Sym^d (V \otimes V^* \otimes V)^{H} \Big)^{\aS_d}$
and observing that the resulting vectors are still linearly independent.
Clearly we can also switch the roles of $\la$ and $\nu$, so part (1) is proved.
We proceed completely analogously for part (2) and (3).
\end{proof}

Next, in order to prove the positivity of $\sm$ we need some positivity results for particular symmetric and skew-symmetric Kronecker coefficients.

Let $\la^t$ denote the partition corresponding to the Young diagram of $\la$ reflected on the main diagonal.
For example, $(5,4,4)^t=(3,3,3,3,1)$.
Partitions that satisfy $\la = \la^t$ are called \emph{self-conjugate}.
Using character theory it is easy to show that $g(\pi,\la,\la^t) = 1$ for $\pi = |\la| \times 1$.
Using eq.~\eqref{eq:g=sk+ak} we know that only one of two cases can occur: Either
$\sk(\pi,\la) = 1$ and $\ak(\pi,\la)=0$
or
$\sk(\pi,\la) = 0$ and $\ak(\pi,\la)=1$.
Theorem~\ref{thm:signaction} below tells us in which case we are.

For a self-conjugate partition $\la$ we consider the number of boxes that are not on the main diagonal of its Young diagram.
Since $\la$ is self-conjugate, this number is even. Half of them are above the main diagonal and half of them below.
For a self-conjugate partition define its \emph{sign} $\sgn(\la)$ to be 1 if the number of boxes above the main diagonal is even, $-1$ otherwise.

\begin{theorem}\label{thm:signaction}
Let $\pi=(D \times 1)$ and let $\la \partinto D$ be self conjugate. Then
$\sk(\pi,\la) = 1$ and $\ak(\pi,\la)=0$ if $\sgn(\la)=1$, and $\sk(\pi,\la) = 0$ and $\ak(\pi,\la)=1$ if $\sgn(\la)=-1$.
\end{theorem}
This is proved in Section~\ref{sec:skcof} using the tableaux basis for the irreducible representations of the symmetric group $\aS_D$.

\medskip
We now consider the positivity of $\sm$ and show that it is positive for almost all cases. First, we prove it when $\lambda$ is a single column. When $\lambda$ has more columns we apply the semigroup property to the sum of its columns to derive positivity.

Set $X_s :=\{ 2,3,4,7,8,12\}$ and $X_a :=\{ 1,2,5,6,10,14 \}$, as the next statement shows these are exactly the sets of exceptional column lengths, for which $\sm$, respectively $\am$, is 0. 

\begin{proposition}\label{prop:columns}
 Let $\ell := \max\{\lfloor \sqrt{a} \rfloor+2, 12\}$. We have that $\sm(1^a,\ell ) >0$ if and only if $a \not \in X_s$ and $\am(1^a,\ell) >0$ if and only if $a \not \in X_a$.
\end{proposition}

\begin{proof}

First, a direct calculation shows that $\sm(1^a,\ell)=0$ for $a \in X_s$ and $\am(1^a,\ell)=0$ for $a \in X_a$.

To prove positivity, we apply Theorem~\ref{thm:signaction}.
For each $a$ we will find self-conjugate partitions $\mu,\nu \vdash a$, such that $\ell(\mu),\ell(\nu) \leq \ell$ and $\sgn(\mu)=1$, $\sgn(\nu)=-1$. Then $1= \sk(1^a,\mu) \leq \sm(1^a,\ell(\mu)) \leq \sm(1^a,\ell)$, and $1 = \ak(1^a,\nu) \leq \ak(1^a, \ell(\nu)) \leq \ak(1^a,\ell)$. We consider three separate cases: $a \leq 14$, $a \in [15,99]$ and $a \geq 100$.

For $a \leq 14, a \not \in X_s$ the corresponding $\mu$ partitions are $(1)$, $(3,1,1)$, $(3,2,1)$, $(5,1^4)$, $(5,2,1^3)$, $(4,3,3,1)$, $(7,1^6)$, $(7,2,1^5)$ and for $\nu$ we have $(2,1)$, $(2,2)$, $(4,1^3)$, $(4,2,1^2)$, $(3,3,3)$, $(6,1^5)$, $(6,2,1^4)$, $(5,3,3,1,1)$.

Let $a \geq 100$ and set $b := \lfloor \sqrt{a} \rfloor$, which is the maximal possible diagonal length in a self-conjugate partition of $a$ (and only if $2 | a-b$). Note that $b \geq 10$. 
Let $r:=a-b^2$, we have that $r=2r_1+c_1$, where $c_1=0,1$ is the residue of $r$ modulo 2. If $c_1=0$, since $(b+1)^2>a$, we have that $r_1 \leq b$, and if $c_1=1$ we have $(b+1)^2> 2r_1+1 +b^2$, so $b>r_1$. 

Let first $c_1=0$, and consider the partitions $\alpha := (b^b + 1^{r_1}, r_1)\vdash a$ with $\ell(\alpha) =b+1$, and $\beta := ( b^{b-2} + 1^{r_1}+1^2, b-2,b-2, r_1,2)$ for $r_1\leq b-2$ or $\beta:= ( (b+1)^{b-2} +1^{r_1 - b+4}, b-2,b-2,b-2, r_1-b+4)$ for $r_1 >b-2$ (note that $r_1 \leq b$, so $r_1-(b-2) \in \{1,2\}$).  We have that $\beta \vdash a$, $\ell(\beta) \leq b+2$ and both $\alpha=\alpha^t$ and $\beta=\beta^t$. We also have that $\alpha$ has $\frac12( a - b)$ boxes above the diagonal, and $\beta$ has $\frac12(a -(b-2) ) = \frac12(a-b) +1$ boxes, so exactly one of $\alpha$ and $\beta$ is odd, and one even, and these are our $\nu$ and $\mu$, respectively. 

Let now $c_1=1$, and set $d:=b-1$. Let $\gamma := ( (d+1)^d + 1^{r_1}+1, d,r_1,1)$, which is self-conjugate since $r_1 \leq b-1=d$, and $\gamma \vdash a$, $\ell(\gamma) \leq d+3 =b+2$. Let $\delta := ( (d+1)^{d-2} + 1^{r_1} +1^5, d-2,d-2,d-2, r_1,5)$ (sorting the last 2 parts $5,r_1$ in decreasing order if $r_1\leq 4$  ) if $r_1 \leq d-2$, and set $\delta:= ( (d+2)^{d-2} + 1^{r_1 - d +7 }, d-2,d-2,d-2,d-2, r_1-d+7)$ if $r_1 >d-2$ (note again that $r_1-d+7 \leq b-1 -d+7 =7\leq d-2$ since $b \geq 10$). We have that $\delta \vdash a$, $\ell(\delta) \leq d+4=b+3$, and $\delta = \delta^t$. Moreover the number of boxes above the diagonal of these partitions is $ \frac12 (a -d)$ and $\frac12(a-d+2)=\frac12(a-d)+1$, so again one is even and one odd, and we set them to $\mu$ and $\nu$ respectively.

Finally, when $a \leq 99$, so $b\leq 9$, we treat the cases as above, noting that the problematic places arise when $c_1=1$ and some of the inequalities $r_1 -d+7 \leq d-2$  or $5 > d-2$ fails. In these cases we replace the problematic $1^{r_1-d+7}$ or $1^5$ by thicker partitions with at most $12 - (d+2) =11-b$ parts. 
\end{proof}

\begin{proposition}\label{prop:small_values}
Let $\lambda$ be a partition of length $\ell \leq 14$ and $\lambda\not\in \{ (1^r): r\in X_s\} \cup \{(2,1,1), (3,1,1), (2,1^7)\}$. Then $\sm(\lambda,7)>0$.
\end{proposition}

\begin{proof}
We use a program written by Harm Derksen and adjusted by Jesko H\"uttenhain that was already used to generate the computational data in \cite{Ike:PhDthesis}.
A direct computation for partitions $\lambda$ with $\ell(\lambda)\leq 12$ and $\lambda_1 \leq 3$ shows that $\sm(\lambda,7) >0$ except for the cases listed above. We also verify that $\sm(\lambda,7)>0$ for the partitions with $\ell(\lambda)=13,14$ and $\lambda_2 =1,2$  or $\lambda_1=4$ and $\ell (\lambda) \leq 4$. 
If $\ell(\lambda)=13,14$ and $\lambda_1 = 3$, then the cases when the second 2 columns of $\lambda$ form one of the exceptional partitions listed above, we have $|\lambda| \leq 23$ and we check by direct computation. Otherwise the second 2 columns have positive $\sm$ and adding them to the first column by the semigroup property we have $\sm(\lambda,7)>0$.

Let $\lambda_1 \geq 4$ and set $c:= \lambda_1 - \lambda_2$ (the number of singleton boxes). 
Let first $c \neq \lambda_1-1$, i.e. $\lambda_2 \neq 1$. If $\lambda_2=0$, then since $\sm((1),1) = \sk( (1), (1) ) =1$ the semigroup applied $c$ times gives $\sm((c),1) >0$. (One can also observe that since $(c)$ is the trivial representation of $\aS_c$, we have $S^2[(c)] =[(c)]$ and so $\sk((c),(c))=1$.)
 If $\lambda_2 \neq 0$, so $\lambda_2 \geq 2$, we can write $\lambda_2 = 2j $ or $\lambda_2= 2j+ 3$. Then we can write $\lambda = (c)+ \sum_i \alpha^i$, where $\alpha^i$ are partitions with all columns longer than 1 and at least two columns each: let $\alpha^i$ consist of the $2i+1,2i+2$ columns of $\lambda$ and $\alpha^{j+1}$ is the last 3 nonsingleton columns if $\lambda_2=2j+3$. Since the calculation showed that all partitions of 2 or 3 columns, each of lengths $\in [2,14]$ have positive $\sm$, we have $\sm(\alpha^i, 7)>0$. Since $\sm( (c), 7)>0$, the semigroup property for $\sm$ gives $\sm(\lambda,7)>0$.
 
 In the case when $c =\lambda_1 -1$ we must have $\lambda=1^k + (c)$. The calculation showed that $\sm(1^k + (3), k)>0$ and since $c = \lambda_1 -1 \geq 3$, by the semigroup property for $1^k+(3)$ and $(c-3)$ we have $\sm(\lambda,7)>0$.
\end{proof}

Finally, we consider the positivity of the classical Kronecker coefficients, as they are needed to derive $\sm$ positivity in some other exceptional situations.

\begin{proposition}\label{cor:column_2}
We have that at least one of the two quantities is positive: $\sm( (2,2,1^a), \ell)>0$ or $\am((2,2,1^a),\ell)>0$, where $\ell =\max\{7, \lceil \sqrt{a+2} \rceil \}$
\end{proposition}
\begin{proof}
Let $\lambda=(2,2,1^a)$. Using equation~\eqref{eq:g=sk+ak} for the Kronecker coefficients
we have that 
$$\sum_{\mu: \ell(\mu) \leq \ell} g(\lambda,\mu,\mu) = \sm(\lambda,\ell)+\am(\lambda,\ell)$$
Let $\mu$ be a self-conjugate partition of $a+4$ and length at most $\ell$, as constructed for example in the proof of Proposition~\ref{prop:columns}. Then Corollary~\ref{cor:kron_2_columns} applies and $g(\lambda,\mu,\mu)$ is strictly positive, implying that at least one of $\sm$ and $\am$ above must also be positive.  
\end{proof}

\begin{proposition}\label{prop:sm_positive}
Let $\lambda$ be  partition of length at most $L$ and $\lambda \not \in \{ (1^2),(1^3),(1^4), (1^7), (1^8), (1^{12}), (2,1^2), (3,1^2), (2,1^7)\}$ and also $\lambda \neq (2,2,1^{k})$ for any $k$. Let $\ell :=\max\{ \lceil \sqrt{L} \rceil +2,12\}$. Then $\sm(\lambda, \ell) >0$.
\end{proposition}

\begin{proof}

Let $X:=(2^{a_2},3^{a_3},4^{a_4},7^{a_7},8^{a_8},12^{a_{12}})$ be the multiset of columns in $\lambda$ which are of the exceptional lengths $X_s$, and let $\beta$ be the partition formed by them. Let $x :=a_2+ a_3+a_4+a_7+a_8+a_{12}$, and let $\alpha$ be the partition formed by the nonexceptional columns of $\lambda$, so $\lambda=\alpha+\beta$. By Proposition~\ref{prop:columns} we have that each column $1^k$ in $\alpha$, $\sm(1^k,\ell)>0$ and so by the semigroup property adding these columns we get  $\sm(\alpha,\ell)  >0$.

Suppose that $x \geq 2$ or $x=0$.  By Proposition~\ref{prop:small_values}  we have that $\sm(\beta,7)>0$. Thus, by the semigroup property we have that $\sm(\lambda,\ell) = \sm(\alpha+\beta,\ell)>0$. 

Suppose for the rest of the proof that $x=1$, so there is exactly one column of length $r\in X_s$. Since $\lambda$ is not one of the exceptional partitions, it must have at least one more column $k$ and since $x=1$, we must have $k \not \in X_s$. 

Let first $r \neq 2$, then $r \not \in X_a$.
Suppose that $k \not \in X_a$ as well. By Proposition~\ref{prop:columns} we have $\am(1^k,\ell)>0$ and by the $\am$ semigroup property we have $\sm(1^k+1^r, \ell)>0$. The remaining columns of $\lambda$ are $\not \in X_s$, so also have positive $\sm$, and we can add them all to obtain $\sm(\lambda,\ell)>0$ by the $\sm$-semigroup.
If $k \in X_a$, then $k\leq 14$ and so $\sm(1^k + 1^r,\ell)>0$ by Proposition~\ref{prop:small_values}. 

Let now $r=2$. Since $\lambda \neq (2,2,1,1,\ldots)$, there must be at least 2 other columns, say of lengths $k_1, k_2 \not \in X_s$. If $k_i \leq 14$ for some $i$, then $\sm(1^{k_i}+1^r,\ell)>0$ by Proposition~\ref{prop:small_values} and adding this partition to the remaining nonexceptional columns we get $\sm(\lambda,\ell)>0$ by the semigroup. If $k_i >14$, then by Proposition~\ref{cor:column_2} at least one of the following holds:
\begin{itemize}
\item  $\sm(1^r+1^{k_1},\ell)>0$: then adding the remaining nonexceptional columns of $\lambda$ by the semigroup property we get $\sm(\lambda,\ell)>0$.
\item $\am(1^r+1^{k_1},\ell)>0$: then since $k_2 \not \in X_a$, we also have $\am(1^{k_2},\ell)>0$, so by the semigroup property we get $\sm(1^r+1^{k_1}+1^{k_2},\ell)>0$. Adding the remaining nonexceptional columns of $\lambda$ we have $\sm(\lambda,\ell)>0$.
\end{itemize}

This exhausts all cases and completes the proof. 
\end{proof}
We can now derive the proof of Proposition~\ref{pro:pos}.
\begin{proof}[Proof of Proposition~\ref{pro:pos}]
This follows directly from Proposition~\ref{prop:sm_positive}: 
since $\lambda_1 \geq 3$, we have that $\lambda$ is not a partition of $1$ or $2$ columns, and since $\lambda \vdash dm \geq 10$, we have that $\lambda$ is not any of the exceptional partitions. We have that  $\ell(\lambda)\leq m^2=L$, and thus $\ell=\max\{m+2,12\}=m+2\leq n$. So $\sm(\lambda,n)\geq \sm(\lambda,\ell)>0$ by Proposition~\ref{prop:sm_positive}.
\end{proof}

\section{Stabilizer-invariants in the Schur modules}\label{sec:stabinv_semigroup}
In this section, we prove Thm.~\ref{thm:sm}.

We introduce the notation that we will need in this section. Let $E$ be a vector space of dimension $n$, let $E^*$ be its dual space. Define $V = E^* \otimes E = \End(E)$. We have that $\Pow^m_n \in S^m V$ is defined by $\Pow^m_n (X) = \tr(X^m)$ for any $X \in V^*$. For any two vector spaces $W,W'$ and any invertible linear map $f : W \to W'$, $f^{-T} : W^* \to W'^*$ denotes its transpose inverse.

We are interested in the stabilizer of $\Pow^m_n$ in $\GL(V)$, that is
$
\calS := \{ g\in \GL(V) \mid g \cdot \Pow^m_n = \Pow^m_n \}.
$
It is characterized by the following theorem.
\begin{theorem}\label{thm: stabilizer of Pow}
If $n,m \geq 3$, The stabilizer of $\Pow^m_n$ in $\GL(V)$ is
\[
\calS = (\PGL(E) \times \langle \omega_m \cdot \Id_V \rangle) \rtimes \langle \tau \rangle
\]
where $\PGL(E) = ad(\GL(E))$ is the image of the adjoint representation $ad: \GL(E) \to \GL(V)$, $\omega_m$ is a primitive $m$-th root of $1$ and $\tau :V \to V$ is defined via $\tau : E^* \otimes E \to E^* \otimes E$, $\eta \otimes e \mapsto \delta^{-1} (e) \otimes \delta(\eta)$, where $\delta : E^* \xto{\sim} E$ is a vector space isomorphism identifying a basis of $E$ with its dual basis.
\end{theorem}

The proof of Theorem~\ref{thm: stabilizer of Pow} is given in Section~\ref{sec:stab}. Denote $\calS_0 = \PGL(E) \subseteq \calS$.

Let $\pi$ be a partition, $\pi \partinto{d}$ with length $\ell(\pi) \leq n^2$. The space of $\calS$-invariants in the Schur module $\bbS_\pi V$ will be determined in two steps. First, we will determine the space of $\calS_0 \times \langle \omega_m \Id_V \rangle$ invariants in $\bbS_\pi V$: this space is $0$ if $d$ is not a multiple of $m$ and it is the space of $\calS_0$-invariants, $\left[ \bbS_\pi V \right]^{\calS_0}$, if $d$ is a multiple of $m$. Afterwards, we determine the space of $\langle \tau \rangle$-invariants $ \left[[\bbS_\pi V]^{\calS_0}\right]^{\langle \tau \rangle}$.

It is immediate that, if $d$ is not a multiple of $m$, then $\bbS_\pi V$ does not contain non-zero invariants, because $\omega_m \Id_V$ acts on $\bbS_\pi V$ by multiplication by $\omega_m^d$, that is $1$ if and only if $d$ is a multiple of $m$.

We proceed as in \cite{BuIk:GCT_and_tensor_rank} to determine the space of $\calS_0$-invariants.

\begin{proposition}
Let $\pi \partinto[n^2]d$. Then $[ \bbS_\pi V ]^{\calS_0} = \sum_{\substack{\mu \partinto{d} \\ \ell(\mu) \leq n}} g(\pi,\mu,\mu)$.
\end{proposition}
\begin{proof}
$\calS_0$ is the image of $\GL(E)$ in $\GL(V)$ via the adjoint representation, so the $\calS_0$-invariant subspace in $\bbS_\pi V$ coincides with the $\GL(E)$-invariant subspace. We have the following decomposition under the action of $\GL(E)$ (see e.g. \cite[Sec. 4.4]{Ike:PhDthesis}):
\begin{equation}\label{eqn: split with kron. spaces}
 \bbS_\pi V = \bbS_\pi (E^* \otimes E) = \bigoplus_{\mu,\nu \partinto{d}} K_\pi^{\mu,\nu} \otimes \bbS_\mu E^* \otimes \bbS_\nu E,
\end{equation}
where $ K_\pi^{\mu,\nu} $ is a multiplicity space whose dimension is the Kronecker coefficient $g(\pi,\mu,\nu)$. In particular, the action of $\GL(E)$ on $K_\pi^{\mu,\nu}$ is trivial. Moreover, if $\ell(\mu) > n$ or $\ell(\nu) >n$, then $\bbS_\mu E^* \otimes \bbS_\nu E = 0$.

If $\ell(\mu),\ell(\nu) \leq n$ then, by Schur's Lemma, it is immediate that $\left[ \bbS_\mu E^* \otimes \bbS_\nu E \right]^{\GL(E)} = 0$ if $\mu \neq \nu$ and $\left[ \bbS_\mu E^* \otimes \bbS_\mu E \right]^{\GL(E)} = \langle \Id_{\bbS_\mu E} \rangle$. We obtain
\begin{equation}\label{eqn: GL(E)-invariants in SpiV}
[ \bbS_\pi V ] ^{\calS_0} = [ \bbS_\pi V ] ^{\GL(E)}  = \bigoplus_{\substack{\mu,\nu \partinto{d} \\ \ell(\mu) ,\ell(\nu) \leq n}} K_\pi^{\mu,\nu} \otimes [\bbS_\mu E^* \otimes \bbS_\nu E]^{\GL(E)} = \bigoplus_{\substack{\mu \partinto{d} \\ \ell(\mu) \leq n}} K_\pi^{\mu,\mu} \otimes \Id_{\bbS_\mu E}.
\end{equation}
The dimension of this space is $\sum_{\substack{\mu \partinto{d} \\ \ell(\mu) \leq n}} g(\pi,\mu,\mu)$.
\end{proof}

In order to determine the space of $\langle \tau \rangle$-invariants in $[\bbS_\pi V]^{\calS_0}$, we will study the action of $\tau$ on the right-hand side of \eqref{eqn: split with kron. spaces}. We follow the discussion of \cite[Sec. 5.2]{BuLaMaWe:Ov_math_iss_in_GCT}. 

If $W$ is a vector space of dimension $n$ and $\lambda$ is a partition $\lambda \partinto{d}$, $\ell(\lambda) \leq n$, then, by Schur-Weyl duality
$
 \bbS_\lambda W = \Hom_{\frakS_d} ( [\lambda], V^{\otimes d}),
$
where $[\lambda]$ is the Specht module associated to $\lambda$.

Given partitions $\pi,\mu,\nu \partinto{d}$, by definition of the Kronecker coefficient
$
 K_\pi^{\mu,\nu} = \Hom_{\frakS_d} ([\pi] , [\mu] \otimes [\nu]).
$

For every $\pi,\mu,\nu$, the following $\GL(E)$-equivariant map realizes a summand on the right-hand side of \eqref{eqn: split with kron. spaces} as submodule of $\bbS_\pi (E^* \otimes E)$:
\begin{equation}\label{eqn: kron. summands embedding}
\begin{aligned}
  K_\pi^{\mu,\nu} \otimes \bbS_\mu E^* \otimes \bbS_\nu E &\to \bbS_\pi (E^* \otimes E) \\
  \phi \otimes \alpha \otimes \beta &\mapsto (\alpha \otimes \beta) \circ \phi
\end{aligned}
\end{equation}
where we use the reordering $ (E^* \otimes E)^{\otimes d} \simeq E^{*\otimes d} \otimes E^{\otimes d}$ (maintaining the relative order of the copies of $E$ and of the copies of $E^*$).

Notice that the isomorphism $\delta: E^* \xto{\sim} E$ induces a vector space isomorphism $E^{* \otimes d} \xto{\sim} E^{\otimes d}$ and that restricts to $\bbS_\lambda E^* \xto{\sim} \bbS_\lambda E$ for every $\lambda \partinto{d}$. Similarly, the map $\tau \in \GL(V)$ acts on $(E^* \otimes E)^{\otimes d}$: its action commutes with the action of $\frakS_d$, so it passes to the components $\bbS_{\pi}(E^* \otimes E)$. More precisely, if $\psi \in \bbS_{\pi}(E^* \otimes E) = \Hom_{\frakS_d}([\pi], (E^*\otimes E)^{\otimes d})$ then $\tau(\psi) = \tau^{\otimes d} \circ \psi$, that is the composition
\begin{equation}\label{eqn: tau tensor d on Spi}
 \begin{aligned}
 {[\pi]} \xrightarrow{\phantom{a}\psi\phantom{a}} (E^*\otimes E)^{\otimes d} &\xrightarrow{\tau^{\otimes d}} (E^*\otimes E)^{\otimes d} \\
 \otimes_j (\beta^j \otimes u_j)& \longmapsto \otimes_j (\delta^{-1}(u_j) \otimes \delta(\beta^j)),
 \end{aligned}
\end{equation}
for $\beta^j \in E^*,u_j \in E$.

For every $\pi,\mu,\nu$, there is an isomorphism $\sigma^\pi_{\mu,\nu} : K_\pi^{\mu,\nu} \to K_\pi^{\nu,\mu}$ obtained via the composition of an element $\psi$ with the canonical isomorphism $[\mu] \otimes [\nu] \simeq [\nu] \otimes [\mu]$; in particular $\sigma^\pi_{\mu,\nu}$ is the inverse of $\sigma^\pi_{\nu,\mu}$ and $\sigma^\pi_{\lambda,\lambda}$ is an element of order $2$ acting on $K_\pi^{\lambda,\lambda}$.

Consider the diagram 
\[
\xymatrix@R-1pc{
 K^{\mu,\nu}_\pi \otimes \bbS_{\mu} E^* \otimes \bbS_{\nu} E \ar@<.5ex>[r] \ar[d]& \bbS_\pi (E^* \otimes E) \ar@<.5ex>[l] \ar[d]\\
 K^{\nu,\mu}_\pi \otimes \bbS_{\nu} E^* \otimes \bbS_{\mu} E \ar@<.5ex>[r] & \bbS_\pi (E^* \otimes E) \ar@<.5ex>[l]\\
 }
\]
where the horizontal arrows from left to right are the $\GL(E^*) \times \GL(E)$-equivariant embeddings as in \eqref{eqn: kron. summands embedding}, the horizontal arrows from right to left are the corresponding projections, the vertical arrow on the right is the $\tau^{\otimes d}$ as in \eqref{eqn: tau tensor d on Spi} and the vertical arrow on the left is the map sending $\phi \otimes \alpha \otimes \beta \in K^{\mu,\nu}_\pi \otimes \bbS_{\mu} E^* \otimes \bbS_{\nu} E$ to 
$
 \sigma_\pi^{\mu,\nu} (\phi) \otimes ( (\delta^{-1})^{\otimes d} \circ \beta) \otimes (\delta^{\otimes d} \circ \alpha).
$
A straightforward calculation shows that the diagram commutes.

In particular, the action of $\tau$ restricts to the summands of \eqref{eqn: split with kron. spaces} where $\mu = \nu$ as
\begin{equation}\label{eqn: restriction of tau to GLE decomp}
\begin{aligned}
 K_\pi^{\mu,\mu} \otimes \bbS_\mu E^* \otimes \bbS_\mu E &\to  K_\pi^{\mu,\mu} \otimes \bbS_\mu E^* \otimes \bbS_\mu E \\
 (\alpha \otimes \beta) \circ \phi &\mapsto ((\delta^{-1} \circ \beta) \otimes (\delta \circ \alpha)) \circ \sigma_\pi^{\mu,\mu}(\phi)
\end{aligned}
\end{equation}

\begin{lemma}\label{lemma: tau restricts to K's}
 The action of $\tau$ restricts to the $\GL(E)$-invariant subspace in $K^{\mu,\mu}_\pi \otimes \bbS_\mu E^* \otimes \bbS_\mu E$.
\end{lemma}
 \begin{proof}
The $\GL(E)$-invariant subspace is $K_\pi^{\mu,\mu} \otimes \Id _{\bbS_\mu E}$. In particular, we need to show that $\tau^{\otimes d} (\Id _{\bbS_\mu E}) =  \Id _{\bbS_\mu E}$ (up to scale). But this is clear as $\tau$, by definition, preserves $\Id _E$ and so $\Id_{\bbS_\mu E}$.
 \end{proof}

 Now, we can conclude
 \begin{theorem}\label{thm:mult}
If $\pi \partinto{d}$, and $d$ is a multiple of $m$, then the space of $\calS$-invariants in $\bbS_\pi V$ is
\[
\left[ \bbS_\pi V \right]^\calS = \textstyle\bigoplus_{\substack{\lambda \partinto{d} \\ \ell(\lambda) \leq n}} sK^{\lambda,\lambda}_{\pi} \otimes \Id_{\calS_\lambda E}
\]
where $ sK^{\lambda,\lambda}_{\pi} = \Hom_{\frakS_d} ([\pi], S^2 [\lambda])$. In particular, its dimension is $\sm(\pi, n) = \sum_{\mu \vdash d, \ell(\mu) \leq n} sk(\pi,\mu)$.
\end{theorem}
\begin{proof}
The entire $\bbS_\pi V$ is invariant under the cyclic group $\langle \omega _m \rangle \subseteq \calS$, therefore the space of $\calS$-invariants in $\bbS_\pi V$ coincides with the subspace of $\langle \tau \rangle$-invariants in $[\bbS_\pi V]^{\calS_0}$.

Restricting to the space of $\GL(E)$-invariants, from Lemma \ref{lemma: tau restricts to K's}, $\tau$ acts on each summand $K^\pi_{\lambda,\lambda} \otimes \Id_{\bbS_\mu E}$ as in \eqref{eqn: restriction of tau to GLE decomp}.

We deduce $[\bbS_\pi V ] ^{\calS} = \bigoplus \left[K^{\lambda,\lambda}_\pi \right]^{\langle \tau \rangle} \otimes \Id_{\bbS_\lambda E}$, where the direct sum ranges over $\lambda \partinto{d}$ with $\ell(\lambda) \leq n$.

The space of $\langle \tau \rangle$-invariants in $K_\pi^{\lambda,\lambda}$ is the space of $\frakS_d$-equivariant maps $[\pi] \to [\lambda]\otimes [\lambda]$ that are fixed by the permutation of the two factors $[\lambda]$. The module $[\lambda] \otimes [\lambda]$ splits under the action of $\tau$ as $[\lambda] \otimes [\lambda] = S^2[\lambda] \oplus \Lambda^2 [\lambda]$. Hence
$
 K^\pi_{\lambda,\lambda} = \Hom_{\frakS_d} ([\pi], [\lambda]\otimes [\lambda]) = \Hom_{\frakS_d} ([\pi],S^2[\lambda]) \oplus \Hom_{\frakS_d} ([\pi],\Lambda^2[\lambda]).
$
The space $sK^{\lambda,\lambda}_\pi = \Hom_{\frakS_d} ([\pi], S^2[\lambda])$ is the invariant subspace under the action of $\tau$ and by definition its dimension is $ sk(\pi,\mu)$.
\end{proof}

Theorem \ref{thm:sm} follows from Theorem \ref{thm:mult} via Peter-Weyl Theorem, as explained in Section \ref{section: introduction}.

\section{Proof of the stabilizer Theorem~\protect\ref{thm: stabilizer of Pow}}\label{sec:stab}

In this section $m,n \geq 3$. Fix a basis $e_1 \vvirg e_n$ of $E$ and its dual basis $\eta^1 \vvirg \eta^n$. Write $x^i_j = \eta^i \otimes e_j \in E^* \otimes E = V$. The expression of $\Pow^m_n$ in coordinates is 
\[
 \Pow^m_n = \sum_{i_1 \vvirg i_m} x^{i_1}_{i_2} x^{i_2}_{i_3}\cdots x^{i_m}_{i_1} \in S^m V.
\]
Write $\xi^j_i$ for the dual basis of $x^i_j$: we can identify $\xi^j_i$ with the differential operator $\frac{\partial}{\partial x^i_j}$.

If $G$ is a group and $H$ is a subgroup, we denote by $N_G(H) = \{g \in G : g Hg^{-1} = H \}$ and $C_G(H) = \{ g \in G : \forall h \in H \ g h g^{-1} = h \}$, respectively, the normalizer and the centralizer of $H$ in $G$. For a group $G$, let $\Aut(G)$ denote the group of automorphisms of $G$. There is a natural group homomorphism $G \to \Aut(G)$, given by $h \mapsto (\phi_h :  g \mapsto h gh^{-1})$; the kernel of this homomorphism is $Z(G)$, the center of $G$; the image of this homomorphism is denoted by $\Inn (G)$, the group of inner automorphisms of $G$. $\Inn(G)$ is a normal subgroup of $\Aut(G)$: let $\Out(G) = \Aut(G) / \Inn (G)$ be the quotient group, called the group of outer automorphisms of $G$.

The stabilizer $\calS$ of $\Pow^m_n$ inherits the Zariski topology of the space $\End(V)$; let $\calS_0$ denote the connected component of the identity in $\calS$.

In this section we prove Thm.~\ref{thm: stabilizer of Pow}. First, we state the following standard fact:

\begin{lemma}[e.g. \cite{Ges:Geometry_of_IMM}, Lemma 2.1]\label{lemma: discrete part normalizes}
 Let $f \in S^d W$ be a polynomial and let $G$ be a connected Lie group acting on $W$. Let $G_f$ be the stabilizer of $f$ in $G$ and let $G^0_f$ be the connected component of the identity in $G_f$. Then $G_f \subseteq N_G (G_f^0)$.
\end{lemma}

Applying Lemma \ref{lemma: discrete part normalizes} to $f = \Pow^m_n$ (in the setting of the lemma we have $G = \GL(V)$, $W=V$, $G_f = \calS$ and $G_f^0 = \calS_0$), we deduce
\begin{equation}\label{eqn: S0 in S in NS0}
\calS_0 \subseteq \calS \subseteq N_{\GL(V)} (\calS_0).
\end{equation}

The outline of the proof is as follows: first we will determine the connected subgroup $\calS_0$ of $\calS$; the second step is determining $N_{\GL(V)} (\calS_0)$ that will be obtained by studying the action of $N_{\GL(V)} (\calS_0)$ on $\calS_0$ via conjugation; finally we will determine $\calS$ exploiting its action on $\calS_0$ via conjugation.

The following observation is important to determine the connected subgroup $\calS_0$.

\begin{observ}\label{obs: identity component from annihilator}
Let $f \in S^dW$. Let $G$ be a connected Lie group acting on $W$ and let $\frakg$ be the Lie algebra of $G$; let $\mathfrak{ann}_\frakg(f) := \{ L \in \frakg : L . f = 0\}$. Then $\mathfrak{ann}_\frakg(f)$ is the Lie algebra of $G_f^0$ and $G_f^0$ is the unique connected subgroup of $G$ with this property. See \cite[Sec. 1.2]{Pro:LieGroups} for details.
\end{observ}

The subgroup $\calS_0$ will be given by the image of the adjoint representation of $\GL(E)$, that is the homomorphism $ad: \GL(E) \to \GL(E^* \otimes E)$, defined by
\begin{align*}
 ad(g) : E^* \otimes E &\to E^* \otimes E \\
 \eta \otimes e &\mapsto g^{-T}(\eta) \otimes g(e).
\end{align*}
The kernel of $ad$ is the center of $\GL(E)$ and its image is denote by $\PGL(E) \subseteq \GL(E^* \otimes E)$. 

\begin{proposition}
The subgroup $\PGL(E) \subseteq \GL(V)$ coincides with $\calS_0$.
\end{proposition}
\begin{proof}
Let $Ad : \End(E) \to \End(V)$ be the differential of $ad$. 

We will prove that $\mathfrak{ann}_{\End(V)} (\Pow^m_n) = \Im (Ad)$. Observation \ref{obs: identity component from annihilator} and the universality of the exponential map (see e.g. \cite[Prop. 3.28]{Hall:LieGroups}) will allow us to conclude that $\calS_0 = ad(\GL(E)) = \PGL(E)$. 

If $L \in \End(E)$, via Leibniz rule, we have 
\[
Ad(L) = -L^T \otimes \Id_E + \Id_{E^*} \otimes L \in \End(E) \otimes \End(E^*) \simeq \End(V). 
\]

It is useful to determine this image in terms of the basis $x^i_j$ and its dual basis.

 The identification $\End(V) \simeq \End(E^*) \otimes \End(E)$ is made explicit via the reordering isomorphism, as follows:
\begin{align*}
 \End(V) &= V^* \otimes V = (E^* \otimes E)^* \otimes (E^* \otimes E) \simeq \\
 &\simeq E \otimes E^* \otimes E^* \otimes E \underset{swap \ E^*}{\simeq} (E \otimes E^*) \otimes (E^* \otimes E) = \End(E^*) \otimes \End(E).
\end{align*}

Therefore, if $\eta^k \otimes e_j \in E^*\otimes E = \End(E)$, we have
\[
 Ad(\eta^k \otimes e_j) = - e_j \otimes \eta^k \otimes \left(\textsum \eta^i \otimes e_i\right) + \left( \textsum e_i \otimes \eta^i \right) \otimes \eta^k \otimes e_j
\]
as an element of $\End(E^*) \otimes \End(E)$; under the reordering isomorphism we obtain
\begin{align*}
 Ad(\eta^k \otimes e_j) &= \textsum_i (e_i \otimes \eta^k \otimes \eta^i \otimes e_j - e_j \otimes \eta^i \otimes \eta^k \otimes e_i) = \\
&= \textsum_i ( \xi^k_i \otimes x^i_j - \xi^i_j \otimes x^k_i).
\end{align*}

The image of $\Pow^m_n$ under the action of basis element $L = \xi^i_j \otimes x^k_\ell$ is given by
\[
L \cdot \Pow^m_n =  x^k_\ell \cdot \frac{\partial}{\partial x^j_i} \Pow^m_n = m \ x^k_\ell \ (X^{m-1})^i_j  .
\]

We will exploit the form of the monomials in this expression: in general, we have
\[
 (X^p)^i_j = \sum_{i_1 \vvirg i_{p-1}}x^{i}_{i_1} x^{i_1}_{i_2} \cdots x^{i_{p-2}}_{i_{p-1}}x^{i_{p-1}}_j .
\]
We will use the following properties of the monomials occurring in $(X^p)^i_j$:
\begin{enumerate}[(i)]
 \item if $k \notin \{i,j\}$, then $k$ appears as upper index the same number of times that it appears as lower index;
 \item if $k \notin \{ i,j \}$ appears as upper index in one variable, then there is at least another variable (possibly equal) where it appears as lower index (and viceversa);
 \item if $i\neq j$, the index $i$ appears as upper index one time more than the number of times it appears as lower index;
 \item if $i\neq j$, the index $j$ appears as lower index one time more than the number of times it appears as upper index;
 \item if $i$ has only one occurrence as upper index, and $j$ has only one occurrence as lower index, then the variable $x^i_j$ does not appear in the monomial. \end{enumerate}

It is easy to show that if $L \in \Im (Ad)$ then $L \cdot \Pow^m_n = 0$. Now, let $L$ be an element in the annihilator of $\Pow^m_n$.

Consider $4$ indices $(i,j,k,\ell)$. We may assume without loss of generality $i,j,k,\ell \in \{1,2,3,4\}$. In all the following cases, we argue that the coefficient of $\xi^i_j \otimes x^k_\ell$ in $L$ has to be $0$; these short technical proofs are based on the fact that the monomials that we consider in $(\xi^i_j \otimes x^k_\ell) \cdot \Pow^m_n$ can only be generated by the basis element $\xi^i_j \otimes x^k_\ell$:

\begin{itemize}
 \dotitem $(i,j,k,\ell) = (1,2,3,4)$. Notice that $(\xi^1_2 \otimes x^3_4) \cdot \Pow^m_n = x^3_4 (X^{m-1})^1_2$ contains the monomial $x^3_4 (x^1_1)^{m-2} x^1_2$. Suppose $x^3_4 (x^1_1)^{m-2} x^1_2$ occurs in $(\xi^\alpha_\beta \otimes x^\gamma_\delta) \cdot \Pow^m_n$. The possibilities for the pair $(\gamma,\delta)$ are $(1,1),(1,2)$ or $(3,4)$. If $(\gamma,\delta) = (1,1)$ then $x^3_4 (x^1_1)^{m-3} x^1_2 = (X^{m-1})^\alpha_\beta$ for some $\alpha,\beta$, but this provides a contradiction with property (i) above, since there are two lower indices ($4$ and $2$) not having the same occurrences as lower and upper index; if $(\gamma,\delta) = (1,2)$ then $x^3_4 (x^1_1)^{m-2} = (X^{m-1})^\alpha_\beta$, for some $(\alpha,\beta)$, providing a contradiction with property (v), since $3$ and $4$ are the only indices not having matching indices; finally if $(\gamma,\delta) = (3,4)$ then $(x^1_1)^{m-2} x^1_2$ occurs in $(X^{m-1})^\alpha_\beta$ and this is possible only if $(\alpha,\beta) = (1,2)$, providing $\xi^\alpha_\beta \otimes x^\gamma_\delta = \xi^1_2 \otimes x^3_4$. The same argument applies to every case where $i,j,k,\ell$ are distinct.
 
\dotitem $(i,j,k,\ell) = (1,1,3,4)$. Notice that $(\xi^1_1 \otimes x^3_4) \cdot \Pow^m_n = x^3_4 (X^{m-1})^1_1$ contains the monomial $ x^3_4 (x^1_1)^{m-1}$. Suppose $(x^1_1)^{m-1} x^3_4$ occurs in $(\xi^\alpha_\beta \otimes x^\gamma_\delta) \cdot \Pow^m_n$. The possibilities for the pair $(\gamma,\delta)$ are $(1,1)$ or $(3,4)$. If $(\gamma,\delta) = (1,1)$, we obtain a contradiction similarly to the second case in the previous part; if $(\gamma,\delta) =(3,4)$, we obtain $(\alpha,\beta) = (1,1)$ namely $ \xi^\alpha_\beta \otimes x^\gamma_\delta = \xi^1_1 \otimes x^3_4$. The same argument applies to every case where $i = j$ and $i,k,\ell$ are distinct. 

\dotitem $(i,j,k,\ell) = (1,2,3,3)$. Notice that $(\xi^1_2 \otimes x^3_3) \cdot \Pow^m_n = x^3_3 (X^{m-1})^1_2$ contains the monomial $x^3_3(x^1_1)^{m-2}x^1_2$. Suppose $x^3_3(x^1_1)^{m-2}x^1_2$ occurs in $(\xi^\alpha_\beta \otimes x^\gamma_\delta) \cdot \Pow^m_n$. The possibilities for $(\gamma,\delta)$ are $(1,1),(1,2)$ and $(3,3)$. If $(\gamma,\delta) = (1,1)$ then $x^3_3(x^1_1)^{m-3}x^1_2 = (X^{m-1})^\alpha_\beta$: since the lower index $2$ does not occur as upper index we have $\beta = 2$ and since $1$ occurs as upper index once more than as lower index we have $\alpha = 1$, but this provides a contradiction with properties (i) and (ii) above, since the index $3$ only occurs in $x^3_3$. If $(\gamma,\delta) = (1,2)$ then $x^3_3(x^1_1)^{m-3}$ occurs in $(X^{m-1})^\alpha_\beta$: if $(\alpha,\beta) = (1,1)$, we obtain a contradiction with (i) and (ii) as in the previous part; if $(\alpha,\beta) = (3,3)$, we obtain a contradiction with property (v). Finally, if $(\gamma,\delta) = (3,3)$, then $(x^1_1)^{m-2}x^1_2 = (X^{m-1})^\alpha_\beta$ and we obtain $(\alpha,\beta) = (1,2)$, namely $\xi^\alpha_\beta \otimes x^\gamma_\delta = \xi^1_2 \otimes x^3_3$. The same argument applies to every case where $i,j,k$ are distinct and $k = \ell$.

\dotitem $(i,j,k,\ell) = (1,1,2,2)$. This case can be solved similarly to the previous one.

\dotitem $(i,j,k,\ell) = (1,2,1,3)$. Notice that $(\xi^1_2 \otimes x^1_3) \cdot \Pow^m_n = x^1_3 (X^{m-1})^1_2$ contains the monomial $x^1_3 x^1_2( x^2_2)^{m-2}$. Suppose $x^1_3 x^1_2( x^2_2)^{m-2}$ occurs in $(\xi^\alpha_\beta \otimes x^\gamma_\delta) \cdot \Pow^m_n$. The possibilities for $(\gamma,\delta)$ are $(1,2),(2,2)$ and $(1,3)$. If $(\gamma,\delta) = (1,2)$ then $x^1_3 (x^2_2)^{m-2}$ occurs in $(X^{m-1})^\alpha_\beta$, but this easily provides a contradiction. A similar contradiction is obtained if $(\gamma,\delta) = (2,2)$. Therefore $(\gamma,\delta)= (1,3)$ and $(\alpha,\beta) = (1,2)$, namely $\xi^\alpha_\beta \otimes x^\gamma_\delta = \xi^1_2 \otimes x^1_3$.

\dotitem $(i,j,k,\ell) = (1,3,2,3)$. This case can be solved similarly to the previous one.

\dotitem $(i,j,k,\ell) = (1,2,1,2)$. This case can be solved similarly to the previous one.

\dotitem $(i,j,k,\ell) = (1,1,1,1)$. Notice that $(\xi^1_1 \otimes x^1_1) \cdot \Pow^m_n = x^1_1 (X^{m-1})^1_1$ contains the monomial $(x^1_1)^m$. It is clear that this can be obtained only from the element $\xi^1_1 \otimes x^1_1$.
\end{itemize}

After this analysis, we observe that the only basis elements of $\End(V)$ that can have non-zero coefficient in $L$ are $\xi^i_j \otimes x^k_i$ and $\xi^i_j \otimes x^j_\ell$.

Now, suppose $\xi^1_2 \otimes x^2_3$ appears in $L$ (and up to rescaling suppose its coefficient is $1$). We have $(\xi^1_2 \otimes x^2_3) \cdot \Pow^m_n = x^2_3 (X^{m-1})^1_2$, that contains, for instance, the monomials $x^2_3 x^1_\ell  x^\ell_1 (x^1_1)^{m-4} x^1_2$ if $m \geq 4$ and $x^2_3 x^1_\ell  x^\ell_2$ if $m = 3$. An argument similar to the ones used above shows that, for every $\ell$, the monomial $x^2_3 x^1_\ell  x^\ell_1 (x^1_1)^{m-4} x^1_2$ can only appear in $(\xi^1_2 \otimes x^2_3 )\cdot \Pow^m_n$ and in $(\xi^\ell_3 \otimes x^1_\ell) \cdot \Pow^m_n$. Therefore, if the basis element $\xi^1_2 \otimes x^2_3$ appears in $L$ with coefficient $1$, then, for every $\ell$, the basis element $\xi^\ell_3 \otimes x^1_\ell$ appears in $L$ with coefficient $-1$. In particular, for $\ell = 2$, $\xi^2_3 \otimes x^1_2$ appears in $L$ with coefficient $-1$. But an argument similar to the one we just used shows that if $\xi^2_3 \otimes x^1_2$ appears in $L$ with coefficient $-1$ then $\xi^1_\ell \otimes x^\ell_3$ appears in $L$ with coefficient $1$.

We just saw that, if $\xi^1_2 \otimes x^2_3$ appears in $L$, then every term of $Ad( \eta^1 \otimes e_3)$ appears in $L$: this shows that, if $L$ is generated by basis elements of the form $\xi^i _j \otimes x^j_k$ with $i \neq k$, then $L$ is contained in the image of $Ad$.

Finally, suppose $\xi^1_2 \otimes x^2_1$ appears in $L$. In $(\xi^1_2 \otimes x^2_1) \cdot \Pow^m_n$, we obtain monomials of the form $x^2_1 x^1_\ell x^\ell_1 (x^1_1)^{m-4} x^1_2$ (or $x^2_1 x^1_\ell x^\ell_2$ if $m = 3$). We observe that the only other basis elements that can generate this monomial are $\xi^1_\ell \otimes x^\ell_2$ and $\xi^\ell_1 \otimes x^1_\ell$. In the first case, we already saw that $L$ has to contain a term generated by elements in the image of $Ad$. In the second case, we can repeat the argument as we did above, and we observe that $L$ contains all the terms in $Ad(\eta^1 \otimes e_1)$.

This concludes the proof that $\Im (Ad) = \mathfrak{ann}_{\End(V)} (\Pow^m_n)$ and so the proof of the Proposition.
\end{proof}

Recall from \eqref{eqn: S0 in S in NS0} that $\calS_0 \subseteq \calS \subseteq N_{\GL(V)}(\calS_0)$. The next step toward the proof of Theorem \ref{thm: stabilizer of Pow} is to determine $N_{\GL(V)}(\calS_0)$.

We will prove that, as an abstract group,
\begin{equation}\label{eqn: abstract NS0}
 N_{\GL(V)}(\calS_0) \simeq (\PGL(E) \times \bbC^{* \times 2} ) \rtimes \langle \tau \rangle,
\end{equation}
where, $\PGL(E) = \calS_0$, $\bbC^{*\times 2}$ is the centralizer $C_{\GL(V)}(\calS_0)$ and $\tau$ is an element of order $2$ acting on $\calS_0$ as in the statement of Theorem \ref{thm: stabilizer of Pow} and on $\bbC^{*\times 2}$ via $(c_1,c_2) \mapsto (c_1^{-1},c_2^{-1})$.

In order to determine the factors of $ N_{\GL(V)}(\calS_0) $, the following general observation will be useful.

\begin{observ}
If $H \subseteq G$ is a subgroup, then $N_G(H)$ acts on $H$ via conjugation, namely there is a group homomorphism 
\begin{align*}
 N_{G}(H) &\to \Aut(H) \\
 g &\mapsto (\phi_g : h \mapsto g h g^{-1}).
\end{align*}

The kernel of this homomorphism is the centralizer $C_G(H)$. The product subgroup $HC_G(H)$ is the kernel of the composition
\[
 N_{G}(H) \to \Aut(H) \to \Out(H)
\]
where the second map is the projection modulo $\Inn(G)$. In particular $N_G(H) / (HC_G(H))$ is (isomorphic to) a subgroup of $\Out(H)$.
\end{observ}

This allows us to determine $N_{\GL(V)}(\calS_0)$ by determining first its centralizer and then realizing the outer automorphisms of $\calS_0$ via conjugation by an element of $\GL(V)$.

The next result characterizes the centralizer $C_{\GL(V)}(\calS_0)$:

\begin{lemma}
 The centralizer $C_{\GL(V)}(\calS_0)$ is isomorphic, as an abstract group, to $(\bbC ^*)^{\times 2}$.
\end{lemma}
\begin{proof}
Every element $s \in \calS_0 \subseteq \GL(V)$ is a linear map $V \to V$. The space $V = E^* \otimes E \simeq \End(E)$ splits under the action of $\calS_0$ as $V = \bbC \Id_E  \oplus \fraksl(E)$, where $\fraksl(E)$ is the subspace of traceless endomorphisms in $\End(E)$. The fact that $g \in C_{\GL(V)}(\calS_0)$ is equivalent to the fact that $g : V \to V$ is $\calS_0$-equivariant. By Schur's Lemma, $g$ acts by non-zero scalars on the irreducible components of $V$ under the action of $\calS_0$: we conclude $C_{\GL(V)}(\calS_0) = \bbC^* \Id_{\bbC \Id_E } \times \bbC^* \Id_{\fraksl(E)}$.
 \end{proof}

Since $\calS_0 \simeq \PGL(E)$ has trivial center, we have $\calS_0 \cap C_{\GL(V)}(\calS_0) = \{ id_{V}\}$, so $\calS_0 \times C_{\GL(V)}(\calS_0) \subseteq N_{\GL(V)}(\calS_0)$.

Moreover, it is known that, if $n \geq 3$, then $\Out(\calS_0) \simeq \bbZ_2$ and an outer automorphism can be realized as follows. Consider the automorphism of $\SL(E)$ defined as follows:
\begin{align*}
\tilde{ \tau}_0 : \SL(E) &\to \SL(E) \\
      g &\mapsto \delta^{-1}\circ g^{-T}\circ \delta,
\end{align*}
where $\delta: E^* \xto{\sim} E$ is the isomorphism that identifies $\eta^i \mapsto e_i$. It is easy to observe that $\tilde{\tau}_0$ is an isomorphism. If we fix coordinates and we identify $\SL(E)$ with the group of $n\times n$ matrices whose determinant is $1$, then $\tilde{\tau}_0 : A \mapsto A^{-T}$. In particular, it maps the center of $\SL(E)$ to itself and therefore it descends to the quotient, defining an isomorphism
\[
\tau_0 : \PGL(E) \to \PGL(E) 
\]
It turns out that $\tau_0$ is an outer automorphism and that it is unique up to conjugation by an inner automorphism (corresponding to the choice of the identification $\delta$). See \cite[Ch. 3]{Dieud:Geom_Groupes_Classiques} for details.

Now, we can characterize $ N_{\GL(V)} (\calS_0)$.
\begin{proposition}
The normalizer $ N_{\GL(V)} (\calS_0)$ is 
\[
(\calS_0 \times \bbC^{*\times 2} ) \rtimes \langle \tau \rangle.
\]
An element $(c_1,c_2) \in \bbC^{* \times 2}$ acts as $c_1 \Id _{\langle \Id_E \rangle} \times c_2 \Id_{\Id_{\fraksl(E)}}$ and $\tau$ acts via $\tau : \eta \otimes e \mapsto \delta^{-1}(e) \otimes \delta(\eta)$.
\end{proposition}
\begin{proof}
It is straightforward to verify that, $\tau$ is an element of $\GL(V)$ of order $2$ and if $s \in \calS_0$, then $\tau s \tau^{-1} = \tau_0 (s)$.

This proves that $(\calS_0 \times \bbC^{*\times 2} ) \rtimes \langle \tau \rangle \subseteq N_{\GL(V)} (\calS_0)$. Passing to the quotient modulo $\calS_0 \times \bbC^{*\times 2}$, we obtain $\langle \tau \rangle \subseteq \Out(\calS_0)$ and since they both have order $2$ we conclude that they are the same.
\end{proof}

In order to prove Theorem \ref{thm: stabilizer of Pow}, it only remains to determine which elements of $N_{\GL(V)}(\calS_0)$ stabilize $\Pow^m_n$. 
\begin{proof}[Proof of Theorem \ref{thm: stabilizer of Pow}]
Obviously $\calS_0 \subseteq \calS$. The map $\tau \in \GL(V)$ induces the transpose on $E^* \otimes E$ and in particular it stabilizes $\Pow^m_n$, so $\tau \in \calS$.

Finally, let $g = c_1 \cdot \Id_{\bbC \Id_E} + c_2 \cdot \Id_{\fraksl(V)} \in C_{\GL(V)}(\calS_0)$. Suppose $g$ stabilizes $\Pow^m_n$. Since $\Pow^m_n$ is not bi-homogeneous in the groups of variables $\{ x^i_i : i=1 \vvirg n\}$ and $\{ x^i_j : i \neq j\}$, we deduce that $c_1 = c_2 =: c$. This shows that $g \cdot \Pow^m_n = c^m \Pow^m_n$ and $c^m=1$ if and only if $c$ is an $m$-th root of $1$.
 \end{proof}

The polynomial $\Pow^m_n$ is not characterized by its stabilizer $\calS$. But we can characterize the subspace of polynomials that are stabilized by $\calS$.

\begin{proposition}\label{prop: space invariant under S}
  Let $f \in S^m V$. Then $f$ is stabilized by the action of $\calS$ if and only if it is a homogeneous symmetric polynomial of degree $m$ in the eigenvalues of the elements of $V^*$. The space of these polynomials has dimension $\# \{ \gamma \vdash m, \ell(\gamma)\leq n\}$ -- the number of partitions of $m$ in at most $n$ parts. When $n\geq m$ this  is asymptotically $\sim \frac{1}{4m\sqrt{3}} \exp\left( \pi \sqrt{\frac{2m}{3}} \right)$.
 \end{proposition}
\begin{proof}
After fixing coordinates, $V^*$ is identified with the space of $n \times n$ matrices, $f$ is a polynomial in matrix entries and $g \in \PGL(E) \subseteq \calS$ acts via conjugation by any element $S_g \in \GL(E)$ whose image in $\PGL(E)$ is $g$. We will prove that $f$ coincides with a symmetric function of the eigenvalus of the elements of $V^*$ on the dense subset of diagonalizable matrices. Passing to the closure we conclude.

Let $A$ be a diagonalizable matrix in $V^*$, namely there exists $S \in \GL(E)$ such that $D = S^{-1} A S$ is diagonal and its diagonal entries are the eigenvalues of $A$. In particular $f(A) = f(D)$; the eigenvalues of $D$ are the same as the eigenvalues of $A$ and $f$ is a polynomial in the entries of $D$, so $f$ is a polynomial in the eigenvalues of $A$ (and clearly it is homogeneous of degree $m$). Moreover, conjugation by a permutation matrix permutes the diagonal entries of $D$, therefore $f$ is a symmetric polynomial.

Conversely, for $A \in V^*$, denote by $\Sigma_A$ the set of the eigenvalues of $A$. Let $g \in \calS$:  we have that $\Sigma_{gA} = \omega_m' \Sigma_A$, where $\omega'_m$ is an $m$-th root of $1$. A symmetric polynomial of degree $m$ has the same value on $\Sigma_A$ and $\Sigma_{gA}$; in particular $f(A) = f(gA)$.

The space of symmetric polynomials of degree $m$ is spanned by the basis $\{ e_{\alpha} | \alpha \vdash m\}$, where $e_{\alpha} := e_{\alpha_1}e_{\alpha_2}\cdots$ and $e_k(x_1,x_2,\ldots) = \sum_{i_1 < i_2 <\cdots < i_k} x_{i_1}x_{i_2}\cdots x_{i_k}$ are the elementary symmetric polynomials, see e.g.~\cite{FulHar:RepTh}. When the number of variables is $n$ we must have $\alpha_i \leq n$, else $e_{\alpha_i}=0$, and the dimension is given by $\# \{ \alpha \vdash m| \alpha_1 \leq n\}$, via conjugation $\gamma = \alpha^t$, this is equivalent to the number of partitions $\gamma$ with $\ell(\gamma) \leq n$. If $m\leq n$, then we have $\alpha_1 \leq m \leq n$, and there is no further restriction on these partitions. The asymptotics is then  given by the classical formula of Hardy-Ramanujan for integer partitions.  
\end{proof}

\begin{observ}\label{obs:notchar}
If $t_1 \vvirg t_n$ are the eigenvalues of $A \in V^*$, then $\Pow^m_n (A) = t_1^m + \cdots +t_n^m$, that is indeed a symmetric polynomial in $t_1 \vvirg t_n$. Moreover, the argument used in the first part of the proof of Prop.~\ref{prop: space invariant under S} applies to every degree, showing that $f$ is invariant under the action of $\PGL(E)$ if and only if it is a symmetric function of the eigenvalues. In particular, the $k$-th elementary symmetric function of the eigenvalues (namely the coefficients of $t^{n-k}$ in the characteristic polynomial) is stabilized by $\tilde{\calS} = (\calS_0 \times \langle \omega_k \rangle )\rtimes \langle \tau \rangle$; in fact $\tilde{\calS}$ is the entire stabilizer \cite[Thm. 3.4]{LiPierce:Linear_preservers_problems}.
 \end{observ}

\section{Symmetric Kronecker coefficients of columns}\label{sec:skcof}
In this section we prove Theorem~\ref{thm:signaction}.

The irreducible $\aS_D$ representation of type $\la \partinto D$ has a concrete description as follows \cite[p.~110]{fult:97}, see also \cite[Sec.~4.1]{Ike:PhDthesis}.

A \emph{tableau} of shape $\la$ is a filling of the boxes of the Young diagram corresponding to $\la$ with entries $1,2,\ldots,|\la|$.
Let $\mathscr T(\la)$ denote the set of all tableaux of shape $\la$.
Then $\IC^{\mathscr{T}(\la)}$ is a finite dimensional vector space with an action of $\aS_D$.
We will quotient out a linear subspace $K(\la)$ as follows:
\begin{compactitem}
\item Given tableaux $T_1$ and $T_2$ of shape $\la$. Then $T_1 + T_2 \in K(\la)$ if $T_2$ arises from $T_1$ by switching two entries in a column.
This relation is called the \emph{Grassmann relation}.
\item Given a tableau $T$. Then $T + \sum_S S \in K(\la)$, where the sum goes over all tableaux $S$ that arise from $T$ by exchanging for some $j$ and $k$
the top $k$ elements from the $(j+1)$th column with any selection of $k$ elements in the $j$th column, preserving their vertical order.
This relation is called the \emph{Pl\"ucker relation}. Our argument will only need the Grassmann relation.
\end{compactitem}
\begin{theorem}[e.g.~\cite{fult:97}]\label{thm:specht}
For $\la \partinto D$ we have $[\la] \simeq T(\la)/K(\la)$ as $\aS_D$ representations.
\end{theorem}
In the light of Theorem~\ref{thm:specht} we identify $[\la]$ with $T(\la)/K(\la)$.
We will always think of tableaux of shape $\la$ as being representatives of cosets in $T(\la)/K(\la)=[\la]$.
In particular $[\la]$ is generated as a vector space by tableaux of shape $\la$ and
$[\pi]\otimes[\la]\otimes[\la]$ is generated by tensors $T'_1 \otimes T'_2 \otimes T'_3$,
where $(T'_1,T'_2,T'_3)$ is a triple of tableaux of shape $(\pi,\la,\la)$.

The following symmetrization map is the linear projection from $[\pi]\otimes[\la]\otimes[\la]$ onto
its invariant subspace 
$([\pi]\otimes[\la]\otimes[\la])^{\aS_D}$:
\[
P(T) := \tfrac 1 {D!} \sum_{\pi \in \aS_D} \pi(T).
\]
Moreover, since the action of $\aS_D$ and $\IZ_2$ (switching the last two tensor factors) commute,
$P$ can be restricted to $[\pi]\otimes S^2[\la]$ or $[\pi]\otimes \Lambda^2[\la]$
and projects onto
$([\pi]\otimes S^2[\la])^{\aS_D}$
and $([\pi]\otimes \Lambda^2[\la])^{\aS_D}$, respectively.
\begin{lemma}
The set $\{P(T'_1 \otimes T'_2 \otimes T'_3) \mid \text{$T'_1$ of shape $\pi$, $T'_2$ of shape $\la$, $T'_3$ of shape $\la$}\}$
forms a generating set of the invariant space $([\pi]\otimes[\la]\otimes[\la])^{\aS_D}$.
\end{lemma}
\begin{proof}
This immediately follows from the fact that the set $\{T'_1 \otimes T'_2 \otimes T'_3 \mid \text{$T'_1$ of shape $\pi$}, \text{$T'_2$ of shape $\la$}, \text{$T'_3$ of shape $\la$}\}$
is a generating set of $[\pi]\otimes[\la]\otimes[\la]$ and that $P$ is the linear projection onto the $\aS_D$ invariant subspace.
\end{proof}

For a shape $\la$ there is a unique tableau whose entries increase from top to bottom, left to right, columnwise. We call it the \emph{column standard} tableau of shape $\la$. Analogously, for a shape $\la$ there is a unique tableau whose entries increase from top to bottom, left to right, rowwise. We call it the \emph{row standard} tableau of shape $\la$. For example, if $\la=(4,3,1)$, its column standard tableau is $\tiny\young(1468,257,3)$, and the row standard is $\tiny\young(1234,567,8)$. If $\la$ is just a column or a row, then the row standard and column standard tableaux coincide and we call it the \emph{standard} tableau. 

\begin{lemma}\label{lem:project}
Let $\pi=(D \times 1)$ and let $\la$ be self conjugate.
If $T_1$ is the standard tableau of shape $(D \times 1)$, $T_2$ is the row standard of shape $\la$, and $T_3$ is the column standard of shape $\la$, then $P(T_1 \otimes T_2 \otimes T_3) \neq 0$.
\end{lemma}
\begin{proof}
Since $\la$ is self-conjugate we have $g(\pi,\la,\la)=1$.
Since the $P(T'_1 \otimes T'_2 \otimes T'_3)$ for tableaux $T'_1$, $T'_2$, $T'_3$ form a generating set of $([\pi]\otimes[\la]\otimes[\la])^{\aS_D}$,
it follows that there exists a tableau triple $(T'_1,T'_2,T'_3)$ with $P(T'_1 \otimes T'_2 \otimes T'_3)=v\neq 0$.
Note that there are no two elements $a$ and $b$ which appear in the same column of $T'_2$ and at the same time in the same column of $T'_3$. Otherwise, the transposition $\tau = (a \ b)$ fixes all columns of the three tableaux, and after the Grassmann relations it changes the total sign by $(-1)^3$, so $\tau T'_1 \otimes \tau T'_2 \otimes \tau T'_3 = - T'_1 \otimes T'_2 \otimes T'_3$, and since $P(T'_1, T'_2, T'_3) = P(\tau T'_1 , \tau T'_2 , \tau T'_3)$ we must have that they both are 0.
Now start with the first column in $T'_3$: by the consideration above we know that its entries use different columns in $T'_2$.
With a permutation $\sigma$ we move them inside their columns in $T'_2$ to the top row
and using the Grassmann relation we obtain 
$P(T'_1 \otimes \sigma T'_2 \otimes T'_3) = \pm v \neq 0$.
We do the same for the second column in $T'_3$ and continue through all columns,
so that we end up with a tableau $\sigma' T'_2$ in which row $i$ contains exactly the entries from column $i$ in $T'_3$.
Moreover, we still have $P(T'_1 \otimes \sigma' T'_2 \otimes T'_3) = \pm v \neq 0$.
Using a permutation $\pi \in \aS_D$ on all three tableaux simultaneously we rename the entries in $\sigma'T'_2$ to make it row standard:
$P(\pi T'_1 \otimes T_2 \otimes \pi T'_3) = \pm v \neq 0$.
It follows that the Grassmann relation suffices to make $\pi T'_3$ column standard:
$\pi T'_1 \otimes T_2 \otimes \pi T'_3 = \pm \pi T'_1 \otimes T_2 \otimes T_3$.
Now using the Grassmann relation on the first tableau gives
$\pi T'_1 \otimes T_2 \otimes T_3 = \pm T_1 \otimes T_2 \otimes T_3$.
We conclude that $P(T_1 \otimes T_2 \otimes T_3) = \pm v \neq 0$.
\end{proof}
\begin{proposition}\label{pro:switch}
Let $\pi=(D \times 1)$ and let $\la$ be self conjugate. Let $T_1, T_2, T_3$ be as in Lemma~\ref{lem:project}.
Then $T_1 \otimes T_2 \otimes T_3 = \sgn(\la) T_1 \otimes T_3 \otimes T_2$.
\end{proposition}
\begin{proof}
There exists a self inverse permutation $\sigma \in \aS_D$ with $\sigma(T_2)=T_3$ and $\sigma(T_3)=T_2$.
Clearly $\sigma(T_1)=\sgn(\pi) T_1$, because $T_1$ is a column.
This permutation $\sigma$ consists of disjoint transpositions switching boxes above the main diagonal with the corresponding box at the transpose position.
Therefore $\sgn(\sigma)=\sgn(\la)$, which concludes the proof.
\end{proof}

\begin{proof}[Proof of Theorem~\ref{thm:signaction}]
Let $T_1$ be standard of shape $(D \times 1)$, $T_2$ be row standard of shape $\la$, and $T_3$ be column standard of shape $\la$.

Let $\mathcal T := T_1 \otimes T_2 \otimes T_3 + T_1 \otimes T_3 \otimes T_2 \in [\pi]\otimes S^2[\la]$.
Using Prop.~\ref{pro:switch} we conclude that if $\sgn(\la)=1$, then $\mathcal T = 2 T_1 \otimes T_2 \otimes T_3$.
Therefore with Lemma~\ref{lem:project} we see that
$P(\mathcal T) \neq 0$.
Therefore $\sk([\pi],[\la])>0$.

Let $\mathcal T' := T_1 \otimes T_2 \otimes T_3 - T_1 \otimes T_3 \otimes T_2 \in [\pi]\otimes \wedge^2[\la]$.
Using Prop.~\ref{pro:switch} we conclude that if $\sgn(\la)=-1$, then $\mathcal T' = 2 T_1 \otimes T_2 \otimes T_3$.
Therefore with Lemma~\ref{lem:project} we see that
$P(\mathcal T') \neq 0$.
Therefore $\ak([\pi],[\la])>0$.
\end{proof}

\section{Vanishing of plethysm coefficients}\label{sec:vanishingpleth}
In this section we prove Prop.~\ref{pro:plethvanish}.
\begin{proof}[Proof of Prop.~\ref{pro:plethvanish}]
Let $\la \partinto md$ with $\la_1 < m$.
We want to show that $a_\la(d[m])=0$.
An known upper bound for $a_\la(d[m])$ are the so-called \emph{Kostka numbers} $K_{\la,d \times m}$:
\begin{equation}\label{eq:upperboundforpleth}
a_\la(d[m]) \leq K_{\la,d \times m},
\end{equation}
which are quantities for which a classical combinatorial description is known.
We will prove Prop.~\ref{pro:plethvanish} by proving the following stronger statement: If $\la_1 < m$, then $K_{\la,d \times m}=0$.
The upper bound~\eqref{eq:upperboundforpleth} follows for example directly from \cite{Gay:Characters_Weyl_gp_on_zero_wt_sp_of_permutations}, see also the exposition in \cite[Thm.~4.3.8]{Ike:PhDthesis}.

The Kostka numbers have a combinatorial interpretation as follows.
A \emph{semistandard Young tableau} of shape $\la$ and content $\mu$ is a filling of the boxes of the Young diagram of $\la$ with entries $1,2,\ldots,\ell(\mu)$
such that every entry $i$ appears exactly $\mu_i$ times and such that
\begin{compactitem}
\item the entries are strictly increasing in each column from top to bottom and
\item the entries are nondecreasing in each row from left to right.
\end{compactitem}
For example $\tiny\young(1112,22,3)$ is a semistandard Young tableau of shape $(4,2,1)$ and content $(3,3,1)$.
The Kostka number $K_{\la,\mu}$ counts the number of semistandard Young diagram of shape $\la$ and content $\mu$.

Given a partition $\la \partinto md$ with $\la_1 < m$.
We claim that $K_{\la,d \times m}=0$.
Indeed, the pigeonhole principle says that for every placement of $m$ 1s to the boxes of $\la$ we will end up with at least one column containing the number 1 at least twice.
Therefore if $\mu_1 > \la_1$ there is no semistandard Young tableau of shape $\la$ and content $\mu$, so we have $K_{\la,\mu}=0$.
Setting $\mu = d\times m$ and observing that $\mu_1=m$ we conclude that $\mu_1 > \la_1$. Therefore $K_{\la,d \times m}=0$.
\end{proof}

\section{Kronecker positivity}\label{sec:kronpos}

Here we consider the positivity of the Kronecker coefficients when one partition is a 2-row or 2-column, which is used to derive some of the positivity results for $\sm$ in Section~\ref{sec:smpositivity}.

\begin{proposition}\label{prop:kron_2row}
We have that $g((a,b), \nu,\nu)>0$ for all partitions $\nu\vdash a+b$, such that $d(\nu) \geq \sqrt{2b+1} $ and $d(\nu)\geq 7$, where $d(\nu)$ is the Durfee size of $\nu$ (i.e. the length of main diagonal of $\nu$).
\end{proposition}

\begin{proof}
Let $r = d(\nu) \geq 7$. We have that $r^2 \geq 2b+1$. By \cite{PakPan:Strict_unimod} we have that $g((r^2-b,b), r^r, r^r) = p_b(r,r) - p_{b-1}(r,r) >0$, where $p_b(r,r)$ is the number of partitions of $b$ which fit inside the $r \times r$ rectangle. Let $\alpha$ be the partition consisting of columns $r+1,r+2,\ldots$ of $\nu$, and $\beta$ be the partition consisting of rows $r+1,r+2,\ldots$ of $\nu$, so that $\nu = (r^r+\alpha, \beta)$, denote $\gamma = r^r+\alpha$ and $\tau = (r^2-b,b)$. 

By the semigroup property of Kronecker coefficients, we have that
$$g( \tau+(|\alpha|), \gamma,\gamma)= g( ( r^2-b+ |\alpha|,b), r^r+ \alpha, r^r + \alpha) \geq g((r^2-b,b),r^r,r^r)$$
since $g((|\alpha|), \alpha, \alpha)=1>0$. Since the Kronecker is invariant under transposition of two partitions we also have
$$g( \tau + |\alpha|+|\beta|, \gamma^t + \beta^t, \gamma^t + \beta^t) \geq g(\tau+|\alpha|, \gamma^t,\gamma^t) = g(\tau + |\alpha|,\gamma,\gamma)>0$$
as again $g((|\beta|),\beta^t,\beta^t) =1>0$. Finally, we have that $\gamma^t = \gamma^t+\beta^t$, so transposing again gives the desired positivity. 
\end{proof}

\begin{corollary}\label{cor:kron_2_columns}
We have that $g(1^a+1^b, \nu,\nu)>0$ for all $\nu=\nu'$ with $d(\nu) \geq\max\{ 7, \sqrt{2b+1} \}$. 
\end{corollary}
\begin{proof}
 Since $g$ is invariant under transposition of two partitions and $\nu^t=\nu$, $(1^a+1^b)^t = (a,b)$, we have $g(1^a+1^b, \nu,\nu) = g(1^a+1^b, \nu^t,\nu) = g( (a,b), \nu,\nu)>0$ by Proposition~\ref{prop:kron_2row}.
\end{proof}

\bibliographystyle{amsalpha}
\bibliography{completebib}

\end{document}